\documentclass[aps,pra,reprint,amsmath,amssymb,superscriptaddress,nofootinbib,noeprint,longbibliography,floatfix]{revtex4-2}
\pdfoutput=1
\usepackage{amsthm}
\newtheoremstyle{mytheoremstyle} 
    {\topsep}                    
    {\topsep}                    
    {\itshape}                   
    {}                           
    {\bfseries}                   
    {.}                          
    {.5em}                       
    {}  

\theoremstyle{mytheoremstyle}
\newtheorem*{thm*}{Theorem}

\newtheorem*{prop*}{Proposition}
\newtheorem*{def*}{Definition}
\newtheorem{lem}{Lemma}
\newtheorem{prop}{Proposition}
\usepackage{newtxtext}
\usepackage{dsfont}
\usepackage{braket}
\usepackage{xcolor}
\usepackage{microtype}
\definecolor{myciteColor}{rgb}{0.0,0.5,0.23}
\usepackage[colorlinks=true,citecolor=myciteColor,linkcolor=myciteColor,urlcolor=myciteColor]{hyperref}
\usepackage{orcidlink}

\newcommand{\czqf}[0]{$\mathrm{CZ}^\mathrm{qf}$}
\newcommand{\czf}[0]{$\mathrm{CZ}^\mathrm{f}$}

\newcommand{\braid}[0]{$\mathrm{BRAID}$}
\newcommand{\Sf}[0]{$\mathrm{S}^\mathrm{f}$}
\newcommand{\Tf}[0]{$\mathrm{T}^\mathrm{f}$}

\begin{document}
\title{Fault-tolerant fermionic quantum computing}
\author{Alexander Schuckert\orcidlink{0000-0002-9969-7391}${}^*$} %
\affiliation{Joint Center for Quantum Information and Computer Science,
University of Maryland and NIST, College Park, Maryland 20742, USA}
\affiliation{Joint Quantum Institute,
University of Maryland and NIST, College Park, Maryland 20742, USA}
\author{Eleanor Crane\orcidlink{0000-0002-2752-6462}} 
\affiliation{Departments of Physics and Mechanical Engineering, Co-Design Center for Quantum Advantage, Massachusetts Institute of Technology, Cambridge, Massachusetts 02139, USA}
\author{Alexey V. Gorshkov\orcidlink{0000-0003-0509-3421}}
\affiliation{Joint Center for Quantum Information and Computer Science,
University of Maryland and NIST, College Park, Maryland 20742, USA}
\affiliation{Joint Quantum Institute, University of Maryland and NIST, College Park, Maryland 20742, USA}
\author{Mohammad Hafezi\orcidlink{0000-0003-1679-4880}}
\affiliation{Joint Center for Quantum Information and Computer Science,
University of Maryland and NIST, College Park, Maryland 20742, USA}
\affiliation{Joint Quantum Institute, University of Maryland and NIST, College Park, Maryland 20742, USA}
\author{Michael J. Gullans\orcidlink{0000-0003-3974-2987}${}^*$}
\affiliation{Joint Center for Quantum Information and Computer Science,
University of Maryland and NIST, College Park, Maryland 20742, USA}
\def\thefootnote{*}\footnotetext{Correspondence to: schuckertalexander@gmail.com, mgullans@umd.edu}
\date{\today}
\maketitle

\textbf{Simulating the dynamics of electrons and other fermionic particles in quantum chemistry~\cite{Mcardle2020}, materials science~\cite{Bauer2020}, and high-energy physics~\cite{Bauer2023} is one of the most promising applications of fault-tolerant quantum computers. However, the overhead in mapping time evolution under fermionic Hamiltonians to qubit gates~\cite{Abrams1997,Bravyi2002,Derby2021,Chen2024} renders this endeavor challenging~\cite{Stanisic2022, Hemery2023}. We introduce fermion-qubit fault-tolerant quantum computing, a framework which removes this overhead altogether. Using native fermionic operations we first construct a repetition code which corrects phase errors only. Within a fermionic color code, which corrects for both phase and loss errors, we then realize a universal fermionic gate set, including transversal fermionic Clifford gates. Interfacing with qubit color codes we introduce qubit-fermion fault-tolerant computation, which allows for qubit-controlled fermionic time evolution, a crucial subroutine in state-of-the-art quantum algorithms for simulating fermions~\cite{Childs2012,Low_2017,Babbush2018,Kivlichan2020}.  
As an application, we consider simulating crystalline materials~\cite{Babbush2018}, finding an exponential improvement in circuit depth for a single time step from $\mathcal{O}(N)$ to $\mathcal{O}(\log(N))$ with respect to lattice site number $N$ while retaining a fermion-site count of $\tilde{\mathcal{O}}(N)$, implying a linear-in-$N$ end-to-end gate depth for simulating materials, as opposed to quadratic in previous approaches. We also introduce a fermion-inspired qubit algorithm with $O(\mathrm{log}(N)$ depth, but which uses a prohibitive number of additional ancilla qubits compared to fermionic hardware. 
We show how our framework can be implemented in neutral atoms, overcoming the apparent inability of neutral atoms to implement non-number-conserving gates. Our work opens the door to fermion-qubit fault-tolerant quantum computation in platforms with native fermions such as neutral atoms, quantum dots and donors in silicon, with applications in quantum chemistry, material science, and high-energy physics.}

    \begin{figure}[ht]
    \centering
    \includegraphics[width=\columnwidth]{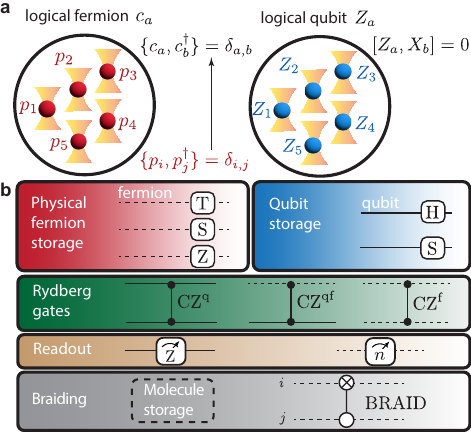}
    \caption{\textbf{Logical fermions and architecture.} \textbf{a}, A logical qubit is encoded in a codeblock formed by many physical qubits. Similarly, we encode a logical fermion into a codeblock formed by many physical fermions. While logical qubit operators commute, logical fermionic operators anti-commute between two codeblocks. Physical fermions can be encoded in fermionic neutral atoms trapped in optical tweezers or lattices (indicated by hour-glass shapes) or in electrons trapped by electric fields in quantum dots or donors in silicon (Methods). \textbf{b}, Neutral-atom architecture, circuit notation, and universal qubit-fermion gate set. Single-fermion gates are performed by detuning the tweezer potential; single-qubit gates by Raman lasers; $\mathrm{CZ}$ gates by Rydberg lasers; braiding gates by photodissociation of molecules as well as inter-tweezer tunneling. In the braiding gate, the site $i$, c.f.~definition in Eq.~\eqref{eq:braid}, is indicated by a cross. We distinguish fermion sites from qubits by using dashed lines and drop the ${}^\mathrm{f}$ superscript in circuits.
    }\label{fig:1}
\end{figure}

Solving the dynamics of many strongly-interacting fermionic quantum particles is central to understanding the dynamics of quarks in the early universe, the reaction dynamics of molecules, and high-temperature superconductivity. While qubit quantum computers promise to help solve these open problems, simulating the evolution of fermions encoded in qubits encounters a large overhead in both the number of two-qubit gates and the circuit depth~\cite{Bravyi2002,Verstraete2005}, which means that a solution of fermionic problems with near-term~\cite{Stanisic2022,Hemery2023} and even fault-tolerant quantum computers~\cite{reiher2017,Babbush2018} is extremely challenging. Quantum simulation experiments with cold fermionic atoms have made major progress in understanding continuum~\cite{greiner2003,bartenstein2004,zwierlein2006,cetina2016,yan2024,vivanco2025} and lattice systems~\cite{cheuk2016a,Mazurenko2017,koepsell2021,xu2025}. These experiments have made steps towards digital computation by trapping individual fermionic atoms in optical tweezers~\cite{Murmann2015,Spar2022}. Additionally, qubits have been encoded into pairs of fermions~\cite{hartke2022}. Quantum simulations using trapped electrons in the solid state have also been demonstrated~\cite{Salfi2016,dehollain2020,Wang2022}. However, turning these fermionic platforms from a tool of scientific exploration into a computational tool requires error-correction and fault-tolerance techniques, for which no practical scheme has been found so far.

Here, we close this gap by introducing a framework for \emph{fermion-to-fermion} fault-tolerant quantum computation. We first prove a theorem that fermionic error correction is not possible using only number-conserving gates, which are most easily implemented in experimental platforms. We then show how to avoid this theorem by using a non-number-conserving gate and concepts from Majorana error correction. This enables us to explicitly show how to perform a universal fermionic gate set including fault-tolerant transversal gates. To enable usage of state-of-the-art quantum algorithms requiring qubit-controlled gates, we introduce qubit-fermion fault-tolerant computation. We apply our framework to time evolution under the Hamiltonian of a crystalline material and study the advantages of our framework compared to all-qubit quantum computation. Finally, we show how to realize these codes and their operations in neutral atoms by removing the previously assumed restriction to number-conserving gates in neutral atoms and introducing the physics behind the implementation of a non-number-conserving braiding gate.

\section*{Fermion-to-fermion error correction gate set}

Our central goal is the construction of an error correcting code that encodes a logical fermion, which we define as a creation operator $c^\dagger_a$ anti-commuting between different codeblocks $a$ and $b$, i.e.~$\lbrace c_a^\dagger,c_b\rbrace=\delta_{ab}$, see Fig.~\ref{fig:1}a. This is by contrast to fermionic operators encoded into a single codeblock of a topological error correcting code~\cite{Landahl2023,Chen2024,McLauchlan2024}, which are difficult to address via transversal gates~\cite{zhou2024} and whose efficiency degrades when applied to fermion models that are not local in 2D. Constructing logical operators that anticommute between different codeblocks is impossible using physical qubits because no qubit operator can be constructed which anticommutes between separate codeblocks. Instead, we use physical fermionic degrees of freedom $p^\dagger_i$ acting on physical sites $i$ and fulfilling $\lbrace p_i^\dagger,p_j\rbrace=\delta_{ij}$ to encode logical fermions. A physical fermion can be realized for example by a fermionic atom trapped with light or an electron in a quantum dot. For concreteness, we will focus on the former implementation and comment on a quantum dot implementation in the Methods. As a second crucial ingredient, we interface these fermions with qubits, which are encoded in two internal energy levels of an atom.

We first consider a fermion-number-conserving gate set acting on many fermionic and qubit sites labeled by indices $i,j$,
\begin{align}
\mathrm{T}^\mathrm{f}_i&=\exp\left( i(\pi/4) n_i\right),\\
\mathrm{CZ}^\mathrm{f}_{ij}&=\exp\left(i\pi n_i n_j\right),\\
\mathrm{CZ}^{\mathrm{qf}}_{ij}&=\exp\left(i\pi n^\mathrm{q}_i n_j\right), 
\label{eq:gateset}
\end{align}
along with the qubit Hadamard gate $\mathrm{H}$, and the qubit phase gate $\mathrm{S}$. For this gate set to be universal (see Methods for a proof), we also need a gate which moves fermions between sites, $\sqrt{\mathrm{iSWAP}}^\mathrm{f}_{ij}=\exp\left(i (\pi/4)\left(p_i^\dagger p_j + h.c. \right) \right)$. For convenience, we also define the gates $\mathrm{S}^\mathrm{f}_i=(\mathrm{T}^\mathrm{f}_i)^2$ and  $\mathrm{Z}^\mathrm{f}_i=(\mathrm{S}^\mathrm{f}_i)^2$ not necessary for universality. $\mathrm{CZ}^{\mathrm{qf}}$ is a qubit-fermion gate dependent on the fermionic number operator $n_j=p_j^\dagger p_j$ and its qubit analogue $n^\mathrm{q}_i=\ket{1}\bra{1}_i$, i.e.~the projector on the $\ket{1}$ state of the qubit. Note that, in principle, we do not require qubits for fermion universality~\cite{Bravyi2010}, but as we will show, being able to couple fermions to qubits will lead to key advantages. Neutral-atom platforms are considered to be constrained to this number-conserving gate set~\cite{GonzlezCuadra2023} as usually the atom number is conserved during gates and each tweezer is initially prepared in a number eigenstate.

The first question is therefore whether these number-conserving operations suffice to encode logical fermions, which is equivalent to  demanding that the codestates are eigenstates of the total fermion number operator. Unfortunately, such states are not sufficient:
\begin{thm*} An error-correcting code using eigenstates of the total fermion number operator as codestates does not have logical fermion operators.
\end{thm*}
The proof (see Methods) relies on the fact that codestates that are eigenstates of the total fermion number operator need to be eigenstates with the same number eigenvalue due to the Knill-Laflamme conditions. This means that logical operators cannot change the total fermion number and therefore need to be even weight in fermionic creation/annihilation operators. Hence, logical operators are bosonic. 

To circumvent this theorem, we use fermion parity eigenstates as codestates, i.e.~states that are superpositions of states that contain either only an odd or only an even number of fermions. A universal gate set on such states requires a non-number-conserving gate. In analogy to the braiding operation performed in Majorana nanowires~\cite{Nayak2008}, we define the braiding gate
\begin{equation}
\mathrm{BRAID}_{ij}=\exp\left(i \frac{\pi}{4} \left(p_i^\dagger- p_i\right) \left(p_j^\dagger+ p_j\right) \right),\label{eq:braid}
\end{equation}
which will also turn out to be a natural logical gate in our codes. 
We propose a scheme to implement this gate in neutral atoms at the end of this work. 

We envision a zoned architecture in which the gates are implemented in separate physical regions of the quantum processor, which has been experimentally realized for qubit-only neutral atom quantum computing in Ref.~\cite{Bluvstein2023}; we specify the zones in which each of the gates is performed in Fig.~\ref{fig:1}b, and introduce our circuit notation. 
 
To define our codes, it is convenient to introduce Majorana fermion operators, which are the real and imaginary parts $\gamma_i$, $\tilde \gamma_i$ of complex fermionic operators, $p_i^\dagger=\frac{1}{2}(\gamma_i + i\tilde\gamma_i)$, fulfilling $\lbrace \gamma_i,\gamma_j\rbrace =2\delta_{ij}=\lbrace \tilde \gamma_i,\tilde \gamma_j\rbrace $ and $\lbrace \gamma_i,\tilde \gamma_j\rbrace =0$. We construct  codes in which we encode two logical Majorana fermions $\gamma^\mathrm{L}$ and $\tilde \gamma^\mathrm{L}$. These are then combined into a logical complex fermion $c^\dagger=\frac{1}{2}(\gamma^\mathrm{L}+i\tilde \gamma^\mathrm{L})$. In order for logical  operators between different codeblocks to fulfill Majorana fermion anti-commutation relations, they each must be a product of an odd number of $\gamma_i$ and $\tilde \gamma_i$, i.e.~they must have odd operator weight. 

\section*{Repetition code for phase errors}

We start with a repetition code which only corrects for phase errors $E=n_i$. This is the dominant error source in neutral atoms~\cite{GonzlezCuadra2023}. As an introductory example, consider the two-site code
\begin{align}
\ket{0}_\mathrm{L} &= \frac{1}{\sqrt{2}}\left(\ket{00}-\ket{11}\right),\label{eq_simple_states1}\\\ket{1}_\mathrm{L} &= \frac{1}{\sqrt{2}}\left(\ket{10}+\ket{01}\right)\label{eq_simple_states2},
\end{align}
where $\ket{11}=p_2^\dagger p_1^\dagger \ket{00}$ and $p_1^\dagger\ket{00}=\ket{01}$ and $\ket{0}$, $\ket{1}$ denotes the absence, presence of a fermion, respectively.
The above code can detect a single phase error, but does not detect loss errors $E=p_i$. The code has two logical Majorana fermion operators $\gamma^\mathrm{L}=\gamma_1$ and $\tilde \gamma^\mathrm{L}=\tilde \gamma_2$, from which we form a logical complex fermion operator by $c^\dagger=\frac{1}{2}\left(\gamma_1 + i \tilde \gamma_2\right)$. This operator acts as $c^\dagger \ket{0}_\mathrm{L} = \ket{1}_\mathrm{L}$, $c^\dagger \ket{1}_\mathrm{L} = 0$. We therefore interpret $\ket{0}_\mathrm{L}$ and $\ket{1}_\mathrm{L}$ as logical fermionic Fock states. 
Two codeblocks $a$ and $b$ each encode such an operator and they fulfill $\lbrace c^\dagger_a, c_b^\dagger\rbrace =0$. Hence, $c^\dagger$ is indeed a logical fermion operator. This simple two-site code is also a stabilizer code, see Fig.~\ref{fig:2}a: its codestates are defined by the operator $i\tilde\gamma_1 \gamma_2$ which acts as the unity operator on both states, $i\tilde\gamma_1 \gamma_2\ket{0}_\mathrm{L}=\ket{0}_\mathrm{L}$ and $i\tilde\gamma_1 \gamma_2\ket{1}_\mathrm{L}=\ket{1}_\mathrm{L}$. This enables generalization to detecting $N-1$ phase errors by using $N$ complex physical sites and the stabilizers $i\tilde\gamma_i \gamma_{i+1}$ for $1<i<N$. For $N\geq 3$, we can also correct phase errors as we show below. This code can also be used as a purely error-detecting code for all errors except atom-loss errors on the edges. In fact, this is the code defined by the ground space of the Kitaev chain, and the two logical Majorana operators $\gamma_1$ and $\tilde \gamma_N$ are its well-known edge states used as qubits in Majorana nanowires~\cite{Nayak2008}. However, by contrast to nanowires, we do not implement the Kitaev Hamiltonian in our system. Hence, we cannot rely on the gap between the ground space of the Hamiltonian and the excited states in order to suppress errors and keep the computation in the ground space of the Hamiltonian. Instead, we require active error correction.

\begin{figure}
    \centering
    \includegraphics{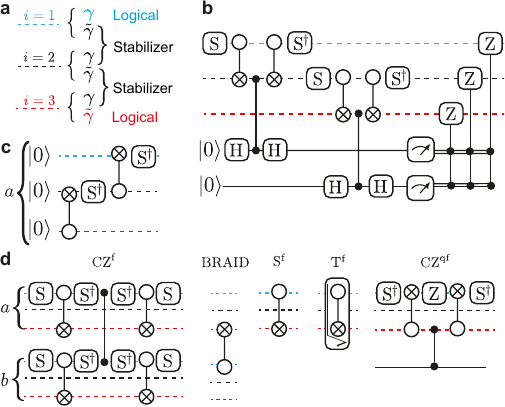}
    \caption{\textbf{Logical operations for a repetition code protecting against phase errors.} \textbf{a}, Illustration of splitting physical fermions for $N=3$ into virtual Majorana fermions. Stabilizers are given by $i\tilde\gamma_1\gamma_2$, $i\tilde\gamma_2\gamma_3$, and the logical fermionic creation operator is given by $c^\dagger=\frac{1}{2}(\gamma_1+i\tilde \gamma_3)$.  \textbf{b}, Syndrome measurement and correction. The two stabilizers are measured by using a Hadamard test on qubits. Conditioned on the measurement outcome, gates are performed to correct for errors. \textbf{c}, Encoding circuit, translated from Ref.~\cite{Mudassar2024}, realized as a chain of braidings. Note that $(\mathrm{S}^\mathrm{f})^\dagger$ can be interpreted as intra-site braiding. \textbf{d}, Logical gates. $\sqrt{\mathrm{BRAID}}_{31}=\exp\left(-\frac{\pi}{8} \tilde \gamma_3 \gamma_1\right)$, which only acts on the two edge atoms. See Methods for a derivation of these circuits.}
    \label{fig:2}
\end{figure}

For implementing active error correction, we need to measure where errors happened and correct for them, which we show how to do in Fig.~\ref{fig:2}b. We use our qubit-fermion gate to map a measurement of the fermion density $n$ to a measurement of the phase of the qubit, which we transform to a measurement of the population of the qubit using Hadamard gates. To measure the stabilizer instead of $n$ on the fermions we conjugate the qubit-fermion gate with two braiding gates as well as phase gates. An ancilla qubit is used for each of the stabilizers. Dependent on the measurement outcome, $\mathrm{Z}^\mathrm{f}$\, gates are performed on the fermions to correct for the errors. This scheme shows the first advantage of using qubits: destructive qubit readout is faster than a non-destructive $n$ measurement~\cite{Endres2019}. To prepare a logical codestate $\ket{0}_\mathrm{L}$, we use a chain of braiding gates---the $\mathrm{S}^\dagger$ gates effectively braid the two Majoranas on the same physical site, Fig.~\ref{fig:2}c.

Next, we discuss how to implement logical gates. We show the circuits in Fig.~\ref{fig:2}d, see Methods for a proof. The basic building block for the logical $\mathrm{CZ}^\mathrm{f}$ is the physical $\mathrm{CZ}^\mathrm{f}$. Similar to the stabilizer measurement, conjugation with braiding gates ``spreads'' the entanglement induced by the physical $\mathrm{CZ}^\mathrm{f}$ to the correct site, with $\mathrm{S}^\mathrm{f}$ gates selecting whether $\tilde \gamma$ or $\gamma$ is acted upon. The braiding gate is simple as it amounts to a physical braiding gate between $\gamma_1$ of one block and $\tilde \gamma_N$ of the other. The $\mathrm{S}^\mathrm{f}$ gate amounts to a braiding gate between the first and $N$th site of a single codeblock. Similarly, the \Tf\,gate is implemented as a braiding operation with half the duration of the \Sf\,gate. Finally, the qubit-fermion \czqf\, is implemented similarly to the $\mathrm{CZ}^\mathrm{f}$ by conjugating the physical \czqf\, with braid and phase gates. 

The utility of this code is based on the fact that, in neutral atoms, phase errors are orders of magnitude more frequent than atom-loss errors. However, because braiding gates convert a single phase error into two correlated atom-loss errors, which cannot be detected, the depth of the circuits which can be performed in practice is still limited by the phase coherence time. Hence, this code is not fault-tolerant even in the absence of atom-loss errors. The fact that logical gates only act on the edge fermions on the other hand reduces the depth of the circuits, as the gate overhead for implementing logical gates is independent of $N$. It might be beneficial to reduce gate depth further by  replacing the \czf gates and BRAIDs in Fig.~\ref{fig:2}d with their arbitrary-angle versions, which turns our gates to arbitrary-angle logical gates. However, for large-scale computations, a truly fault-tolerant code is necessary.

\section*{Fermion-qubit fault tolerance}

We now show how to construct an error-correcting code correcting for both phase errors $E=n$ and atom-loss errors $E=p$ and show how to perform fault-tolerant gadgets. In the context of encoding qubits in Majorana nanowire platforms, Bravyi, Terhal, and Leemhuis~\cite{Bravyi2010} noted that a particular class of qubit error correcting codes, so-called weakly self-dual Calderbank-Shor-Steane (CSS) codes, can be converted into a Majorana stabilizer code by replacing qubit operators in the stabilizers by Majorana fermions. Motivated by the prospect of high error-resilience for codes with odd-weight logical operators, they noted that for codes with an odd number of physical Majorana fermions, a single odd-weight logical operator can be encoded. This operator is given by the total parity, i.e.~the product of all Majorana operators. Our key observation is that, because there is an odd number of physical Majorana fermions, this logical operator is in fact a logical Majorana fermion operator. Even more importantly, while in Majorana nanowire platforms, a complex physical fermion would have to be encoded in two logical Majoranas which are encoded in spatially distinct physical Majoranas, our framework using complex physical fermions enables the encoding of two such Majorana codes, and therefore one complex logical fermion, in the \emph{same} set of physical fermions.  It turns out that a general prescription to constructing complex fermion codes from weakly self-dual CSS codes is to replace the $Z$ operators in the qubit code stabilizers by $\gamma$ operators and the qubit $X$ operators by $\tilde \gamma$ operators, see Fig.~\ref{fig:3}a. The logical Majorana operators are then given by the two parities $\gamma^\mathrm{L}=\prod_{i=1}^N  \gamma_i$ and $\tilde \gamma^\mathrm{L}=\prod_{i=1}^N \tilde \gamma_i$, from which we construct a single logical fermion operator $c^\dagger=\frac{1}{2}\left(\gamma^\mathrm{L}+i\tilde\gamma^\mathrm{L}\right)$. As an example, we consider the family of triangular honeycomb Majorana color codes, which encodes a single logical fermion with odd distance $d$ into $N=(3d^2+1)/4$ physical fermions. In particular, $d=3$ is the fermionic version of the Steane code~\cite{Steane1996}, Fig.~\ref{fig:3}a. Such codes can correct $(d-1)/2$ loss or phase errors. 

\begin{figure}
    \centering
    \includegraphics[width=\columnwidth]{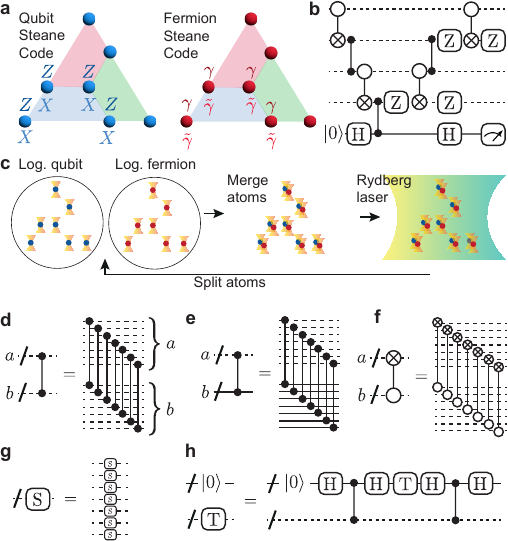}
    \caption{\textbf{Fermion-qubit fault-tolerant quantum computation.} \textbf{a}, Construction of the fermionic Steane code, the smallest member of the family of color codes we consider. Stabilizers are obtained by replacing products of $X$ ($Z$) operators in the qubit code with products of $\gamma$ ($\tilde \gamma$), respectively. Corners on areas with the same color are part of a stabilizer as shown for the bottom left area. \textbf{b}, Mid-circuit measurement of $-\gamma_1 \gamma_2 \gamma_3 \gamma_4$ (counting sites from bottom to top). By inverting the $\mathrm{BRAID}$s (i.e. $\mathrm{BRAID}_{12}\rightarrow \mathrm{BRAID}_{21}$), we measure  $-\tilde \gamma_1\tilde \gamma_2\tilde \gamma_3\tilde \gamma_4$, the other stabilizer. \textbf{c}, Neutral-atom implementation of transversal \czqf\, gates, performed by interleaving a logical qubit and fermion and applying a global Rydberg laser. This implements the circuit shown in subfigure e. \textbf{d}, Transversal \czf. This is implemented by the logical qubit in subfigure c with a logical fermion. \textbf{e}, Transversal qubit-fermion \czqf.  Transversal \textbf{f}, \braid\, and \textbf{g}, \Sf\, gates are similarly implemented by parallel operations. \textbf{h}, \Tf gate. Implemented by using a T gate applied on an ancilla qubit, which is then mapped onto the fermion. All the operations shown here are straightforwardly generalized to codes with larger (odd) $d$, see Methods.}
    \label{fig:3}
\end{figure}
 
Of central utility is the fact that we can use qubits for controlled operations. One example is the measurement of stabilizers using the Hadamard test, as shown in Fig.~\ref{fig:3}b. From these stabilizers, we can reconstruct whether a $\tilde \gamma$ or $\gamma$ error (or both) happened on a fermion. Such single-Majorana errors are corrected by coupling an ancilla site $a$ prepared in $\ket{0}$ to the site with the error, acting with $(\mathrm{BRAID}_{ai})^2=\gamma_i\tilde \gamma_a$ and then measuring the ancilla. The measurement circuit in Fig.~\ref{fig:3}b is made fault-tolerant by for example using the Shor scheme, i.e.~replacing the single ancillary qubit with a set of qubits in a cat state and repeatedly measuring the stabilizers. This stabilizer measurement is also used for fault-tolerant preparation of the logical codestates from a trivial starting state. Note that the correction of stabilizers is optional for both state preparation and error correction---instead of enforcing stabilizers to be +1, their value can instead be tracked~\cite{gottesman1998}. Moreover, we expect that error correction only needs to be applied after each gate, and not during movement, as demonstrated in neutral-atom qubit quantum computers~\cite{Bluvstein2023}.
  
  Magically, this fermion Steane code yields highly intuitive and easy-to-implement fault-tolerant logical gates. Transversal logical gates are implemented, similar to qubit Steane codes, by applying the corresponding physical gate in parallel on all atoms, c.f.~Fig.~\ref{fig:3}c-g, see Methods. This set of gates turns out to be the Majorana Clifford gates \braid, \czf, \Sf, which in analogy to qubit Clifford gates map a single Majorana operator string to another single Majorana string. These transversal Clifford gates are inherently fault-tolerant as an error on one physical codeblock can only spread to one physical fermion in the other codeblock. 

One of the key tools enabled by our approach is the interfacing of our fermionic Steane code with a qubit Steane code (with matching $N$) using the qubit-fermion \czqf\,gate: by applying physical \czqf\,in parallel (see Fig.~\ref{fig:3}e), we realize a logical, transversal \czqf. This enables fault-tolerant qubit-fermion quantum computation. In particular, we will discuss in the next section that this has immense utility for implementing modern quantum algorithms for simulating fermions. 

This gate set of $\mathrm{BRAID}$, \czqf, \czf, \Sf, together with $\mathrm{H}$ and $\mathrm{S}$ on qubits, forms the ``qubit-fermion Clifford group'' as it maps a Majorana-qubit string to a single other Majorana-qubit string. This gate set needs to be combined with a \Tf gate or a qubit T gate to enable universal qubit-fermion computation. While Majorana fermion schemes for T-gate synthesis have been proposed~\cite{Litinski2018}, they are more involved than qubit schemes due to the requirement to create two-fermion magic states. We instead propose using a qubit $\mathrm{T}$ gate applied on a qubit Steane code and transferring it to the fermions using our transversal logical \czf gate, Fig.~\ref{fig:3}h. This therefore enables a significant simplification compared to Majorana $\mathrm{T}$-gate synthesis.

\section*{Logical qubit-fermion quantum simulation}
\begin{figure}
    \centering
    \includegraphics[width=\columnwidth]{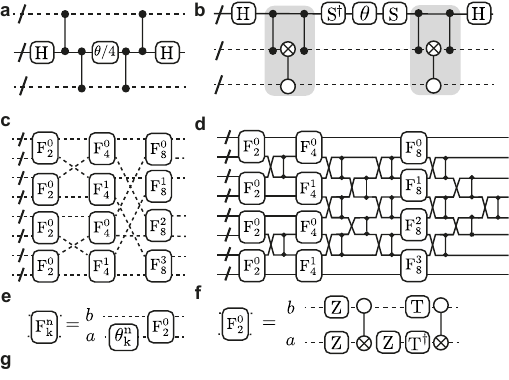}\\
    \begin{tabular}{|c|c|c||c|}
    \hline\multicolumn{1}{|c|}{} & \multicolumn{2}{c||}{Qubits}                     & \multicolumn{1}{c|}{Fermions} \\
    \hline
         & fSWAP~\cite{Kivlichan2018} & pFFFT+pINT  & FFFT+pINT\\
         \hline 
        Depth & $\mathcal{O}(N)$  & $\mathcal{O}(\log{N})$ & $\mathcal{O}(\log{N})$\\
        \hline
        Ancillas & 0 & $\mathcal{O}(N^2)$ + $\mathcal{O}(N\log{N})$ & 0 + $\mathcal{O}(N\log{N})$\\
        \hline
        Volume & $\mathcal{O}(N^2)$ & $\tilde{\mathcal{O}}(N^2)$ & $\tilde{\mathcal{O}}(N)$ \\
        \hline
    \end{tabular}
    \caption{\textbf{Fermion-qubit logical quantum simulation primitives and comparison to the qubit-only approach.} \textbf{a}, Qubit-controlled arbitrary-angle interaction gate. The box with label ``$\theta/4$'' is a qubit gate $\exp(-i(\theta'/2)X)$ with $\theta'=\theta/4$, synthesised from discrete gates using the Solovay-Kitaev algorithm. See Fig.~\ref{fig:1} for our convention of the labeling of the braiding gate. \textbf{d}, Qubit-controlled arbitrary-angle braiding gate. The gray box is the qubit-controlled braiding gate $\exp(-(\pi/4)Z \tilde \gamma_i \gamma_j)$. \textbf{c}, FFFT circuit implemented using logical fermions. Crossing lines indicate swapping, implemented by movement of the  atoms. \textbf{d}, Corresponding FFFT circuit implemented using logical qubits, c.f.~Ref.~\cite{Verstraete2009}. \textbf{e}, Decomposition of the n-mode FFFT into an arbitrary-angle phase gate $\exp(-i\theta n)$ with $\theta=\theta_k^n=2\pi k/n$ and a two-mode FFFT. \textbf{f}, Compilation of the two-mode FFFT into fermionic fault-tolerant gates.  \textbf{g}, Scaling of depth, number of ancillas and circuit volume (defined as the product of depth and the number of ancillas) with the number of sites $N$ for a single Trotter step of the plane-wave basis Hamiltonian in Eq.~\eqref{eq_hubbard}. We compare the lowest-depth previous method~\cite{Babbush2018}, our qubit-only and our fermion-qubit algorithms. ``pINT'' refers to a parallelized interactions scheme using the quantum FFT~\cite{low2019a}, ``pFFFT'' to our version of the FFFT which using the parallelization of CZ circuits introduced in Ref.~\cite{hoyer2005}.}
    \label{fig:4}
\end{figure}
A prime area of application for fermionic quantum computing is simulating electrons in molecules and materials. In second quantized approaches, the continuous-space wavefunctions of electrons are cast into a discrete lattice model through a choice of basis, which leads to a Hamiltonian coupling fermionic operators as $H_\mathrm{qchem}=\sum_{ij}J_{ij}c^\dagger_i c_j+ \sum_{ijkl}V_{ijkl} c^\dagger_i c^\dagger_j c_k  c_l$, with $J_{ij}$ and $V_{ijkl}$ constants that depend on the choice of basis. This generally leads to a number of Hamiltonian terms scaling as $\sim N^4$ with the basis size $N$. For a Trotterized simulation, this implies the same scaling for the number of gates per Trotter step. However, the plane-wave-dual basis~\cite{Babbush2018}, which is particularly suitable for simulating crystalline materials, enables writing the chemistry Hamiltonian as
\begin{align}
    H&= \mathrm{FFFT}^\dagger \sum_i \epsilon_i n_i \mathrm{FFFT} \notag\\&\quad+ \sum_{i<j} U_{ij} \left(n_i-\frac{1}{2}\right) \left(n_j-\frac{1}{2}\right),\label{eq_hubbard}
\end{align}
where $\epsilon_i$, $\mu_i$, and $U_{ij}$ are classically calculable coefficients and $\mathrm{FFFT}$ denotes the fermionic fast Fourier transform ~\cite{Verstraete2009}. This formulation reduces the scaling of the number of terms to $\sim N^2$. Efficient algorithms for dynamics, as well as for ground and finite-temperature state preparation, in addition often require controlling this Hamiltonian on a qubit~\cite{Childs2012,Low_2017}.

Our next goal is therefore to implement qubit-controlled arbitrary-angle evolution under each of the terms in Eq.~\eqref{eq_hubbard} individually. Time evolution under the Hamiltonian $\epsilon_i n_i$ until time $t$ is implemented by replacing the T gate in Fig.~\ref{fig:3}h with an approximate arbitrary-angle rotation of angle $\epsilon_i t$, where $t$ corresponds to time, synthesized from discrete gates with e.g. the Solovay-Kitaev algorithm. The interaction gate $\exp\left(-iZ\frac{\theta}{2}(n_i-\frac{1}{2})(n_j-\frac{1}{2})\right)$ is implemented by the circuit shown in Fig.~\ref{fig:4}a, which maps the arbitrary-angle single-qubit gate onto the fermions using \czqf\, gates. A related gadget shown in Fig.~\ref{fig:3}b realizes an arbitrary-angle qubit-controlled braiding gate $\exp(-\frac{\theta}{2}Z\tilde \gamma_i \gamma_j)$, where again \czqf\, gates map the phase gate from the qubit to a fermion. The \braid s then ``spread'' the gate onto the second fermion. Two such gadgets are combined to yield the qubit-controlled hopping gate via $\exp\left(iJ_{ij}\Delta t Z\left(c_i^\dagger c_j + h.c. \right) \right)=\exp(-\frac{J\Delta t}{2} Z\tilde \gamma_j\gamma_i)\exp(-\frac{J\Delta t}{2} Z\tilde \gamma_i\gamma_j)$. These arbitrary-angle gates also enable time evolution under the more general Hamiltonian $H_\mathrm{qchem}$ by using the compilations in Ref.~\cite{GonzlezCuadra2023}.  
Finally, we need to implement the FFFT, which may be reduced to the 1D FFFT circuit shown in Fig.~\ref{fig:4}c, c.f.~Ref.~\cite{Verstraete2009}. In analogy to the classical fast Fourier transform, the FFFT uses a divide-and-conquer reduction to two-site Fourier transforms $F^k_N$ (Fig.~\ref{fig:4}e, f) whose phases depend on the number of sites $N$ and an index $k$. Despite the fact that it is a 1D transformation, it requires non-local connectivity, which is implemented in neutral atoms by moving the tweezers~\cite{Bluvstein2022}. 

\section*{Advantages over qubit quantum simulation}

We now show that our approach yields a large advantage over currently known qubit algorithms for simulating Trotterized time evolution under the Hamiltonian of crystalline materials.

The most depth-efficient qubit algorithm to date achieves $\mathcal{O}(N)$ depth per Trotter step through the use of fermionic-SWAP (fSWAP=SWAP$\cdot$CZ) networks, where $N$ is the number of fermionic sites in the model~\cite{Kivlichan2018}. Currently-known local encodings of qubits into fermions~\cite{Verstraete2005} yield little advantage for such all-to-all connected hoppings. The FFFT described above can also be used in qubits with fSWAP networks, where atom rearrangements are replaced by layers of fSWAP gates, Fig.~\ref{fig:4}d. Because of the parallelizability restriction imposed by the CZ gates, this leads to gate depth $\mathcal{O}(N\log N)$. We improve the gate depth by using $\mathcal{O}(N^2)$ ancillas and log-depth Clifford-only fanout gates. This enables a parallelization of the CZ gates in the FFFT circuit~\cite{hoyer2005}, see Methods. This ``pFFFT'' method achieves $\mathcal{O}(\mathrm{log}(N))$ depth, but the $\mathcal{O}(N^2)$ ancillas mean that the circuit volume stays $\mathcal{O}(N^2)$ just like for the fSWAP implementation. Furthermore, the ancillas are \emph{logical} ancillas, meaning the overhead in terms of physical qubits for these $\mathcal{O}(N^2)$ ancillas is prohibitive.

By contrast, in our fermionic platform, the fact that fSWAP gates are implemented by atom rearrangements (Fig.~\ref{fig:4}c) means that a depth of $\mathcal{O}(\log N)$ can be implemented without ancillas, yielding $\mathcal{O}(N)$ space and $\tilde{\mathcal{O}}(N)$ volume complexity.  We note that the $N$-scaling of movement time does not prevent the exponential depth reduction from leading to an exponential time saving. This is because the constant overhead of error correction dominates, see Methods. We achieve another depth improvement through our compilation of the constituent FFFT subroutines (see Fig.~\ref{fig:4}e, f), which only takes half the previously-known T-gate depth~\cite{Kivlichan2020}.

At first sight, the logarithmic depth of the kinetic term is not useful because the interaction term requires $\mathcal{O}(N)$ depth because of the parallelizability restriction imposed by the $\mathcal{O}(N^2)$ terms. However, similarly to the qubit version of the FFFT, we can use ancillas to reduce the circuit depth by using the translational invariance of $U_{ij}$. This enables using a quantum circuit of the classical fast Fourier transformation~\cite{low2019a,asaka2020}, see Methods. Because the Fourier transform of $n_i$ is a complex number, we encode it in a binary register with $\mathcal{O}(\log(N))$ ancillas. Therefore, this ``parallelized interactions'' scheme in total uses $\mathcal{O}(N\log(N))$ ancillas, for both qubit and fermion quantum computers. After this transform, the density interactions are diagonal and implementable in depth $\mathcal{O}(\mathrm{log}(N))$. 

Taking both interaction and kinetic terms together, we thus achieve a log-depth Trotter step for materials. For qubit quantum computers, it uses $\mathcal{O}(N^2)$ qubits, but for fermionic quantum computers only $\mathcal{O}(N\log(N))$ fermionic sites. Therefore, the circuit volume is a factor of $N/\log(N)$ lower for fermions. 

The remaining question is how this log-depth Trotter step reflects on the end-to-end complexity of time-evolving up to time $t$. We use the Trotter scheme introduced in Ref.~\cite{childs2021}, which uses $\mathcal{O}(Nt)$ Trotter steps for fixed simulation accuracy and large, fixed Trotter order. Using our log-depth Trotter steps, we therefore find an overall depth complexity of $\tilde {\mathcal{O}}(Nt)$, which is essentially linear in both time and the number of sites. While we show that both qubit and fermion quantum computers can reach this scaling, only the fermion version requires a linear number of modes, while at least within our approach, qubits need a prohibitively expensive quadratic number of qubits using fSWAP networks. 

For the more general $H_\mathrm{qchem}$, at first glance, no depth advantage appears possible, since current state-of-the-art schemes combining double factorization in the molecular basis with fSWAP networks achieve depth $\mathcal{O}(N^2)$ without ancillas~\cite{motta2021} and $\mathcal{O}(N)$ depth with ancillas~\cite{luo2025} for a single Trotter step, independent of the molecule size. By contrast, fermionic quantum computers can reach depth $\mathcal{O}(N)$ for large molecules, see Methods, which to our knowledge is the first algorithm for time evolution under the chemistry Hamiltonian which achieves linear depth per Trotter step without using ancillas. For Hamiltonians which do not allow for a low-rank factorization, such as the sparse Sachdev-Ye-Kitaev model~\cite{xu2020}, we expect such depth speedups to become apparent already at small system sizes.

Within our fault-tolerant gate set, no advantage in the number of T gates is possible if the T-gate count is the only parameter which is optimised for. This can be directly seen by mapping the fermionic gates to qubit gates using the Jordan-Wigner transform and noting that only the \Tf\,gate requires T gates (Methods). We note that steady progress in T-gate synthesis means they do not completely dominate the simulation cost over Clifford gates~\cite{zhou2024}. Moreover, our depth reductions directly translate into a large space-time volume saving compared to qubit platforms. Therefore, we expect advantages in terms of the total runtime of a fermionic simulation.

\begin{figure}
    \centering
    \includegraphics[width=\columnwidth]{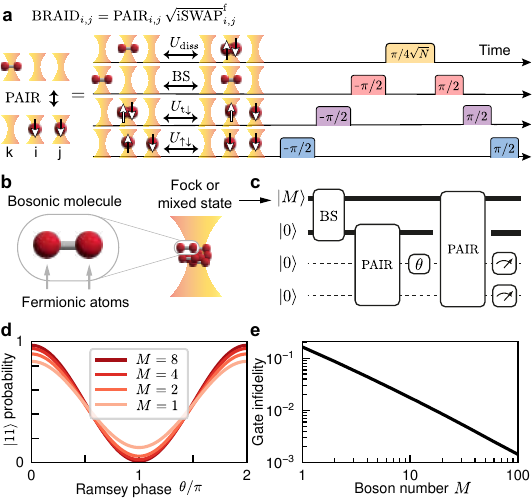}    \caption{\textbf{Implementation of the braiding gate in neutral atoms.} \textbf{a}, Illustration of the sequence of operations to implement an effective braiding gate. Braiding is a product of pairing and the tunneling gate $\sqrt{\mathrm{iSWAP}}^\mathrm{f}=U_{t\downarrow}(\pi/4)$. $U_{\mathrm{t}\downarrow}(\theta)$ and $U_{\uparrow\downarrow}(\theta)$ act only on fermions, while the beamsplitter acts only on molecules, and $U_\mathrm{diss}$ is a fermion-molecule gate. The pairing gate dissociates a bosonic molecule into two fermions on two separate sites. The pulse duration of the dissociation interaction depends on the average number $M=\braket{b^\dagger b}$ of molecules in the tweezer. \textbf{b}, The molecules are formed from the same atomic species as the fermionic atoms used. \textbf{c}, Ramsey sequence to probe the effective coherence of a ``Bell-pair'' $(\ket{00}+i\ket{11})/\sqrt{2}$ created by the effective pairing gate. The phase gate is defined as $\exp(-i\theta n)$. The beamsplitter gate, defined above Eq.~\eqref{eq:pair}, has phase $\theta=\pi/4$. \textbf{d}, Ramsey fringes as a function of initial molecule Fock state number $M$. \textbf{e}, Simulated pairing gate infidelity as a function of molecule Fock state number $M$. See Fig.~\ref{fig:ED}f for the circuit used to determine the infidelity.}
    \label{fig:5}
\end{figure}

\section*{Implementation in neutral atoms}
 
Finally, we describe how to implement our framework in neutral atoms, which requires the implementation of the universal qubit-fermion gateset shown in Fig.~\ref{fig:1}. $\mathrm{CZ}^{\mathrm{qf}}$ and $\mathrm{CZ}^{\mathrm{f}}$ are implemented using a high-lying Rydberg state of the atom, similarly to the routinely implemented qubit-qubit $\mathrm{CZ}^{\mathrm{q}}$~\cite{Jaksch2000,GonzlezCuadra2023}. $\sqrt{\mathrm{iSWAP}}^\mathrm{f}$ is implemented by bringing two tweezers close to each other, inducing tunneling of atoms between the sites~\cite{GonzlezCuadra2023} or by placing the atoms into an optical lattice~\cite{gyger2024} and controllably performing tunneling by using superlattices~\cite{zhu2025,Chalopin2024}. $\mathrm{T^\mathrm{f}}$, $\mathrm{S^\mathrm{f}}$, $\mathrm{Z^\mathrm{f}}$ are implemented by detuning the trapping frequency of the target tweezer or optical lattice site. What is left to show is how to implement $\mathrm{BRAID}$. We propose its realisation as a product $\mathrm{BRAID}_{ij}=\mathrm{PAIR'}_{ij} \mathrm{\sqrt{iSWAP}}^\mathrm{f}_{ij}$, where $\mathrm{PAIR'}_{ij}$ is the pairing gate \begin{equation}
    \mathrm{PAIR'}_{ij}=\exp\left(i\frac{\pi}{4} \left(p_i^\dagger p_j^\dagger +h.c.\right) \right). \label{eq:pairtarget}
\end{equation}

To implement pairing, we propose the pulse-sequence in Fig.~\ref{fig:5}a.  It employs photodissociation of homonuclear diatomic bosonic molecules consisting of the same atomic species as the fermions~\cite{Chin2004,Jiang2011}, see Fig.~\ref{fig:5}b. Photodissociation realizes the boson-fermion gate $U_\mathrm{diss} =\exp\left(-i \frac{\pi}{4\sqrt{M}}  \left(b_i p^\dagger_{i\uparrow}  p^\dagger_{i\downarrow} + b^\dagger_i p_{i\downarrow}  p_{i\uparrow}\right)\right)$,
where $b$ is the bosonic annihilation operator of molecules (we assume there are $M$ of them) residing in the ground state of the tweezer and obeying commutation relations $[b_i,b_j^\dagger]=\delta_{ij}$. $\uparrow$ and $\downarrow$ label two internal states of the atom, where so far we assumed that all atoms reside in $\downarrow$, i.e.~$p_i=p_{i\downarrow}$. In order to convert this same-site gate into a non-local gate, we conjugate it with a fermion spin-selective tunneling $U_{\mathrm{t}\downarrow}(\theta)=\exp\left( -i \theta \left(p^\dagger_{i\uparrow} p_{j\uparrow} +h.c.\right)\right)$ and a fermionic spin-flip $U_{\uparrow\downarrow}(\theta)=\exp\left( -i \theta \left(p^\dagger_{i\downarrow} p_{i\uparrow} +h.c.\right)\right)$. To suppress interactions between the molecules and the fermions, the molecules reside in a separate tweezer when not in use for gates and are tunneled into the target tweezer by moving the two tweezers close to each other, implementing a beam-splitter operation $\mathrm{BS}(\theta)=\exp\left(\theta \left(b_i^\dagger b_j -h.c.\right)\right)$. Altogether, the sequence in Fig.~\ref{fig:5}a implements
\begin{equation}
    \mathrm{PAIR}_{ij}=\exp\left(i\frac{\pi}{4\sqrt{M}} \left(b_k p_i^\dagger p_j^\dagger +h.c.\right) \right).\label{eq:pair}
\end{equation}

We now discuss how we use Eq.~\eqref{eq:pair} to effectively implement Eq.~\eqref{eq:pairtarget}. If the molecules were in a coherent state with large amplitude $\sqrt{M}$, we could replace $b\rightarrow\braket{b}=\sqrt{M}$ and the two gates would be equal. However, in the case of atoms trapped in dipole traps, Fock states or mixtures thereof are more naturally prepared than coherent states. To overcome this conceptual challenge, we in the following use the fact that large-$M$ Fock states cannot be distinguished from coherent states by measurements~\cite{dowling2006}. 

To show that we can indeed effectively implement Eq.~\eqref{eq:pairtarget} employing a large-$M$ Fock state of molecules, consider the thought experiment shown in  Fig.~\ref{fig:5}c. The first pairing gate is supposed to create a ``fermionic Bell-pair'' $(\ket{00}+i\ket{11})/\sqrt{2}$. To probe whether there is coherence, a phase gate with phase $\theta$ is applied and a second $\mathrm{PAIR}$ gate rotates the measurement basis into the effective Bell basis. This realizes a Ramsey-type experiment. Simulating this experiment numerically, we find Ramsey fringes with a contrast that approaches unity as $M\rightarrow \infty$. This shows that, after the first pairing gate, there was indeed effective coherence between states $\ket{00}$ and $\ket{11}$. An intuitive way to understand this observation is that the molecule number distribution in each mode after the  beamsplitter is approximately Poissonian---just like in a coherent state. The entanglement between both modes then leads to an effective phase coherence between the modes. This effective coherence is then revealed by the second pulse, which disentangles the fermions from the molecules. We say ``effective'' coherence because the reduced density matrix of the fermions is always diagonal; indeed, by measuring the total number of molecules, one projects the fermions onto $\ket{00}$ or $\ket{11}$. Therefore, it is important that the molecules are not measured while the fermions are entangled with them. In Fig.~\ref{fig:ED}, we additionally perform state tomography numerically, showing that the state after the first pairing gate is indeed effectively a Bell pair.

We probe the gate fidelity of our effective pairing gate in Fig.~\ref{fig:5}e using Choi state tomography, which we generalize to fermionic gates, see circuit shown in Fig.~\ref{fig:ED}. We find that only $100$ molecules suffice to create a gate with $>99.8\,\%$ fidelity. In the Methods, we also show that even a mixed state with a Poissonian number distribution can be used instead of $\ket{M}$.  Non-trivially, Fig.~\ref{fig:5}c and Fig.~\ref{fig:5}d show that our scheme works even if the second pairing gate is not performed on the same molecule site, which enables full parallelizability of pairing operations. For that to be the case, the molecule tweezers need to be effectively coherent with respect to each other, just like lasers manipulating different qubits within a quantum computer need to be phase-locked. The cross-coherence of molecule tweezers is guaranteed by creating them from one common Fock state with a beamsplitter as done in Fig.~\ref{fig:5}c. To parallelize this further, many molecule tweezers can be created by repeatedly applying beamsplitters starting from one large Fock state, for example by tunneling from a single large dipole trap. While our proposal is challenging to implement, coherent dissociation of Feshbach molecules is routinely performed in cold-atom experiments~\cite{Chin2004}.

\section*{Outlook}
 

Our work opens the field of fermion-qubit error correction and fault-tolerant algorithms. Of prime importance is the error threshold of our codes for the strongly biased errors present in experiments. Moreover, fermionic versions of quantum low-density-parity check codes may be developed to enable a finite rate of logical fermions per codeblock. Due to our reliance on transversal gates, algorithmic fault tolerance can be explored within our approach~\cite{zhou2024}. 


We have so far mainly used the presence of qubits as a helper to implement logical operations. In high-energy physics, the possibility to directly simulate qubit-fermion models enables simulation of fermions coupled to large-spin degrees of freedom in lattice gauge theories~\cite{GonzlezCuadra2023, BOM}. For this purpose, qudits could be directly implemented with large-spin atomic species~\cite{Gorshkov2010,Zache2023}, rotational levels of molecules~\cite{Albert2010}, or an oscillator mode~\cite{Gottesman2001}. Relatedly, the depth-efficiency of our approach might enable simulation of quantum gravity models~\cite{xu2020}.

Boson-fermion models could be studied by combining our schemes with oscillator-qubit quantum simulation~\cite{LGT, BOM} using either motional modes of neutral atoms in tweezers~\cite{Scholl2023, BOM} or bosonic atoms. Suitable boson-to-boson error correction schemes would need to be found~\cite{Lloyd1998,Braunstein1998}. This way, the longstanding problems of the influence of phonon-electron interactions on thermal, structural, and vibrational material properties~\cite{Baroni2001} as well as the interplay of molecular vibrations with the electronic wavefunction in chemistry~\cite{Agostini2022} could be investigated using potentially vastly shallower circuits than with qubit-only platforms. 

Our architecture might offer an advantageous approach not only for simulating fermions but also for quantum computations that are natively qubit-only. We encoded logical fermion operators in odd-weight fermion operators---instead, logical qubits can be encoded in even-weight logical operators, for example Majorana surface codes which have been shown to have high thresholds~\cite{Vijay2015}. One advantage of this approach is that lost atoms are replenished by  coupling to the molecule bath, offering an alternative to continuous atom reloading~\cite{gyger2024}, which is one of the challenging aspects of neutral-atom quantum computing. 

Finally, motivated by our depth improvement of the qubit algorithm for materials, development of fermion-hardware-inspired qubit algorithms might lead to further improvements in qubit algorithms for fermionic simulation.

\section*{Acknowledgments}
We acknowledge discussions with Ignacio Cirac, Daniel Gonzalez-Cuadra, Christopher Kang, Paul Julienne, Robert Ott, Shraddha Singh, Maryam Mudassar, Nathan Wiebe, Torsten Zache, Johannes Zeiher, and Hengyun Zhou.
 
\textbf{Funding.}  This material is based upon work supported in part by the U.S.\ Department of Energy, Office of Science, National Quantum Information Science Research Centers, Quantum Systems Accelerator. Additional support is acknowledged from the following agencies. This work was supported in part by the  Department of Energy (DOE), Office of Science, Office of Nuclear Physics, via the program on Quantum Horizons: QIS Research and Innovation for Nuclear Science (award no. DE-SC0023710) and the National Science Foundation's Quantum Leap Challenge Institute in Robust Quantum Simulation (award no. OMA-2120757). A.V.G.~was also supported in part by the DoE ASCR Accelerated Research in Quantum Computing program (award No.~DE-SC0020312), AFOSR MURI, DoE ASCR Quantum Testbed Pathfinder program (awards No.~DE-SC0019040 and No.~DE-SC0024220), NSF STAQ program, and DARPA SAVaNT ADVENT. E.C.~was also supported in part by ARO MURI (award No.~SCON-00005095), and DoE (BNL contract No.~433702). Support is acknowledged from the U.S.\ Department of Energy, Office of Science, National Quantum Information Science Research Centers, Co-design Center for Quantum Advantage (C2QA) under contract number DE-SC0012704. C2QA participated in this research.   

\textbf{Competing interests.}  All authors declare no competing interests.

\bibliography{bib}
\clearpage

\setcounter{thm}{0}
\setcounter{prop}{0}
\setcounter{figure}{0}
\renewcommand{\thefigure}{S\arabic{figure}}
\setcounter{table}{0}
\renewcommand{\thetable}{S\arabic{table}}
\setcounter{equation}{0}
\renewcommand{\theequation}{S\arabic{equation}}
\begin{center}
    \Large{Methods}
\end{center}
\normalsize

\section*{Proofs}
\subsection*{Proof of Theorem}

In this section, we present the details behind the proof of the Theorem.

We first discuss some preliminaries.
Error correction and detection works by asking whether the computation is still in the code space or not. For an error to be detectable, it must therefore either act in the same way on all code states or make the state leave the codespace. In the case where an the error acted differently on the code states, the detection would conclude that there is no error (because the computation is still in the code space) and therefore that error is not detectable. The Knill-Laflamme condition~\cite{Knill1997,Eastin2009} formulates this condition mathematically. It ensures that errors do not send codestates to states that have an overlap with other codestates. If this condition were not fulfilled, then the detection of an error could lead to a projection onto the wrong codestate, which is an uncorrectable logical error. Formally, the condition states that, for an error operator $E$ to be detectable, $PEP \propto P$ needs to be fulfilled, where $P$ is the projector on the codespace, i.e.~the space spanned by the logical codestates $\ket{0}_\mathrm{L}$ and $\ket{1}_\mathrm{L}$,  representing, respectively, the absence and presence of a logical fermion. 
\begin{lem}
 Consider an error-detecting code which detects phase errors $E_j=n_j$. If this code uses codestates that are eigenstates of the total number operator in a single codeblock, $N_\mathrm{tot}=\sum_j n_j$, then all code states must have the same number eigenvalue.
\end{lem}
\begin{proof}
    The Knill-Laflamme condition for phase errors $E=n_j$ implies $P N_\mathrm{tot} P \propto P$. Furthermore, having codestates that are number eigenstates implies $(1-P)N_\mathrm{tot} P =0$. Taking these two conditions together, we find $N_\mathrm{tot} P \propto P$. This implies that all codestates must have the same number eigenvalue.
\end{proof}
\begin{def*} Fermionic logical operators are operators that anti-commute between two codeblocks.
\end{def*}
We note that this definition does not exclude fermionic logical operators that anti-commute within one codeblock~\cite{McLauchlan2024}. However, our definition is more restrictive - logical operators that anti-commute within one codeblock may commute between two codeblocks. This is for instance the case when encoding two anti-commuting logical operators in a qubit surface code~\cite{Landahl2023}.
\begin{lem}
     Fermionic logical operators have odd weight in physical Majorana operators.
\end{lem}
\begin{proof}
    First, we define new Majorana operators $\eta$ such that $\eta_{2i}=\gamma_i$ and $\eta_{2i+1}=\tilde \gamma_i$. We then write the logical operators $O_a$ and $O_b$ on two codeblocks as a product of $\eta$ Majorana operators. Two such products of Majorana operators   fulfill~\cite{Bravyi2010} \begin{equation}
        O_a O_b = (-1)^{|A|\cdot |B|+|A\cap B|} O_b O_a,\label{eq:maj_comm}
    \end{equation}
    where $A$ and $B$ are the supports on which $O_a$ and $O_b$ are defined (i.e.~on the lattice on which the $\eta$ operators are defined). 
    Because $O_a$ and $O_b$ are defined on two codeblocks, their support does not overlap,  $|A\cap B|=0$. Therefore, they can only fulfill fermionic anti-commutation relations if $|A|$ and $|B|$ are odd, i.e. if the two operators have odd weight.
\end{proof}
\begin{thm*} An error correcting code using number eigenstates as codestates does not have fermionic logical operators.
\end{thm*}
\begin{proof} First, consider stabilizer codes using number eigenstates. In this case, by Lemma 1,  the total fermion parity $\pm\exp(i\pi N_\mathrm{tot})$ is a stabilizer because all the codestates have the same number eigenvalue, which is an integer. As noted in proposition 6.1 in Ref.~\cite{Bravyi2010}, this implies that there cannot be an odd-weight logical operator. To repeat the argument in Ref.~\cite{Bravyi2010}, this is because all logical operators need to commute with the stabilizers and hence the total fermion parity. This is only possible if they are even weight in Majorana operators, which means they are not fermionic logical operators by Lemma 2. 

More generally, i.e. not assuming a stabilizer code,  consider two logical operators $O^\mathrm{L}_a$ and $O^\mathrm{L}_b$ acting on two codeblocks $a$ and $b$. For them to be logical fermionic operators, they need to fulfill anti-commutation relations when acting on the codespace, i.e. $\lbrace O^\mathrm{L}_a,O^\mathrm{L}_b\rbrace P_aP_b=0$, where $P_a$ and $P_b$ are the projectors on the code spaces of $a$ and $b$. To see whether this is the case, consider $O^\mathrm{L}_aO^\mathrm{L}_bP_aP_b$. Because of Lemma 1, the total fermion parity in codeblock $a$ commutes with $P_a$. Hence, the operator decomposition of $P_a$ in terms of Majorana operators must only contain operator strings with even weight. Eq.~\eqref{eq:maj_comm} then implies that $P_a$ commutes with any fermion operator in codeblock $b$, in particular the logical ones,  $O^\mathrm{L}_aO^\mathrm{L}_bP_aP_b=O^\mathrm{L}_aP_aO^\mathrm{L}_bP_b$. Note that for $O^\mathrm{L}_a$ to be a logical operator,  $(1-P_a)O^\mathrm{L}_aP_a=0$ must hold~\cite{Eastin2009}, implying $O^\mathrm{L}_aP_a=P_aO^\mathrm{L}_aP_a$. Because $P_aO^\mathrm{L}_aP_a$ only acts on the code space, it commutes with the fermion parity. Hence, its operator decomposition only contains even-weight operators and therefore commutes with any operator in codeblock $b$. This implies $O^\mathrm{L}_aP_aO^\mathrm{L}_bP_b=O^\mathrm{L}_bO^\mathrm{L}_aP_aP_b$. Hence, $[O^\mathrm{L}_a,O^\mathrm{L}_b]P_aP_b=0$ and logical operators are not fermionic.
\end{proof} 

\subsection*{Proof of absence of T-gate advantage}

In this subsection, we present a proposition proving that there is no T-gate advantage.

\begin{prop}
    Fermion-to-fermion error correcting codes with transversal \braid, \czf, \Sf gates and non-transversal \Tf gate do not have a lower minimal \Tf-gate count than the T-gate count of the equivalent qubit circuit. 
\end{prop}
\begin{proof}
    We show this by converting the fermionic circuit to a qubit circuit using the Jordan-Wigner encoding. The $\mathrm{CZ}^\mathrm{f}$ gate maps to a qubit CZ, which can be compiled to qubit Clifford gates, whereas $\mathrm{S}^\mathrm{f}$ and $\mathrm{T}^\mathrm{f}$ gates directly map to S and T gates, respectively. \braid\,gates map to multi-qubit gates, but their compilation is possible using only Clifford gates. Hence, it is always possible to construct a qubit circuit with exactly the same number of T gates as the fermion circuit.
\end{proof}

\subsection*{Proof of universality of the discrete qubit-fermion gate set}

In this subsection, we show the universality of our discrete qubit-fermion gate set.

\begin{prop}
      \braid, \czf, \czqf, $\mathrm{H}$, $\mathrm{S}$ and T (or \Tf), are universal for qubit-fermion quantum computation.
\end{prop}
\begin{proof}
We would like to do universal qubit-fermion quantum computation on a system with $L_\mathrm{q}$ qubits and $L_\mathrm{f}$ fermions. To do this, we introduce $2L_\mathrm{q}$ fermionic ancillas and encode $L_q$ qubits into dual-rail fermion qubits, i.e.~$\ket{0}\rightarrow \ket{01}$, $\ket{1}\rightarrow \ket{10}$. We then swap the initial qubit state into those dual-rail fermion qubits (we discuss how to do so below). We then perform the qubit-fermion circuit entirely using the universal fermionic gate set \braid, \czf, and \Tf ~\cite{Bravyi2010} on the $L_\mathrm{f}+2L_\mathrm{q}$ fermions (where \Tf gates are implemented using additional qubit ancillas and T,  H and \czqf gates as shown in Fig.~\ref{fig:3}h). At the end of the circuit, we swap the information back from the fermionic ancillas into the qubits.

It remains to show how to implement dual-rail-qubit SWAP gates. A SWAP gate is decomposed into three CNOT gates of two types, one where the qubit is the control and one where the dual-rail qubit is the control. If the dual-rail qubit is the control, a CNOT is simply $\mathrm{H} \mathrm{CZ}^\mathrm{qf} \mathrm{H}$, where the $\mathrm{CZ}^\mathrm{qf}$ is acting on the first fermion site (because it is in $\ket{1}$ when the dual-rail qubit is in $\ket{1}$). If the qubit is the control, we decompose the CNOT as
\begin{equation}
    e^{i\frac{\pi}{4}\left(1-X^\mathrm{DR}\right)\left(1-Z\right)}.
\end{equation}
Noting that $X^\mathrm{DR}=c^\dagger_1 c_2 +c^\dagger_2 c_1$, we can decompose this gate into a qubit phase gate, a $\pi/4$ hopping gate, and a $\pi/4$ qubit-controlled hopping gate, whose compilations are shown in Fig.~\ref{fig:4}b in terms of the claimed gate set.
\end{proof}

\section*{Compilation of logical operations}

We show that the circuits shown in Fig.~\ref{fig:2} and Fig.~\ref{fig:3} indeed implement the target logical operations.

\subsection*{Useful identities for fermion gates}
In this subsection, we provide useful identities for proving circuit identities. 

First, we summarize in Table \ref{tab:gates} how Majorana and qubit operators transform under our gate set.

\subsubsection*{\texorpdfstring{$\mathrm{BRAID}$}{BRAID} gate identities}
In our discrete gate set, the only braiding operation we consider is $\mathrm{BRAID}_{ij}=\exp\left(-\frac{\pi}{4} \tilde \gamma_i \gamma_j\right)$. For the compilations of logical gates, the braidings $\exp\left(-\frac{\pi}{4} \tilde \gamma_i \tilde \gamma_j\right)$ and $\exp\left(-\frac{\pi}{4} \gamma_i \gamma_j\right)$ are also useful. These can be engineered by
\begin{align}
    \exp\left(-\frac{\pi}{4} \tilde \gamma_i \tilde \gamma_j\right) &= (\mathrm{S}^\mathrm{f}_j)^\dagger \mathrm{BRAID}_{ij} \mathrm{S}^\mathrm{f}_j,\notag \\
    \exp\left(-\frac{\pi}{4} \gamma_i \gamma_j\right) &= \mathrm{S}^\mathrm{f}_i \mathrm{BRAID}_{ij} (\mathrm{S}^\mathrm{f}_i)^\dagger,\notag\\
    \mathrm{BRAID}_{ij}^\dagger &= \mathrm{Z}^\mathrm{f}_j \mathrm{BRAID}_{ij} \mathrm{Z}^\mathrm{f}_j,
    \label{eq:arbbraids}
\end{align}
which can be proven by rewriting
$\mathrm{BRAID}_{ij}=\frac{1}{\sqrt{2}}\left(1-\tilde \gamma_i \gamma_j\right)$
and using the transformations in Table \ref{tab:gates}. Note that $(\mathrm{BRAID}_{ij})^\dagger\neq \mathrm{BRAID}_{ji}$.

\subsubsection*{\texorpdfstring{$\mathrm{CZ}$}{CZ} gate identities}
From $(n^\mathrm{q})^2=n^\mathrm{q}$ and $n^2=n$ it follows that
\begin{align}
    \exp(-i\pi n^\mathrm{q}_in_j)&=\mathds{1}-2n^\mathrm{q}_in_j,\notag\\
    \exp(-i\pi n_i n_j)&=\mathds{1}-2n_i n_j.
\end{align}

\subsection*{Repetition code logical operations}
In this subsection, we show explicitly that the circuits shown in Fig.~\ref{fig:2} realize logical gates in the repetition code. Logical gates need to leave the stabilizers invariant. For the repetition code, the logical Majorana fermions are distinct from those on which the stabilizers act (c.f.~Fig.~\ref{fig:2}a). The logical gate compilations below do not act on the Majorana fermions which are part of the stabilizers and therefore leave stabilizers invariant. What is left to show is that our compilations act as they should on the logical Majoranas. To do so, recall that the two logical Majoranas of a single codeblock are given by single physical Majorana fermion operators on the edges, $\gamma^\mathrm{L}=\gamma_1$ and $\tilde \gamma^\mathrm{L}=\tilde \gamma_N$, and, in particular, $n=\frac{1}{2}(1+i\tilde\gamma^\mathrm{L}\gamma^\mathrm{L})=\frac{1}{2}(1+i\tilde\gamma_N\gamma_1)$. In the following, we will drop the ``L'' superscripts and instead denote the codeblock in the superscript. For example, $\tilde\gamma^a_N$ is the $\tilde\gamma^\mathrm{L}$ logical Majorana operator of codeblock $a$.

\begin{table}
    \begin{tabular}{|c|c|c|c|c|c|c|}
    \hline
        &  $\gamma_i$& $\tilde \gamma_i$& $\gamma_j$& $\tilde \gamma_j $ & $\mathrm{Z}_i$& $\mathrm{X}_i$\\[0.5ex] \hline\hline
        $\mathrm{T}_i^\mathrm{f}$& $\frac{1}{\sqrt{2}}(\tilde \gamma_i+\gamma_i)$& $\frac{1}{\sqrt{2}}(\tilde \gamma_i-\gamma_i)$ & & & &\\
        \hline
        $\mathrm{S}_i^\mathrm{f}$ &  $\tilde \gamma_i$ & $-\gamma_i$& & & &\\
        \hline
        $\mathrm{Z}^\mathrm{f}_i$ & $-\gamma_i$ &$-\tilde\gamma_i$ & & & &\\
        \hline
        $\mathrm{H}_i$ & & & & & $\mathrm{X}_i$& $\mathrm{Z}_i$\\
        \hline
        $\mathrm{Z}_i$ & & & & & & -$\mathrm{X}_i$\\
        \hline
        BRAID$_{ij}$ &  & $-\gamma_j$ &$\tilde\gamma_i$ & & & \\
        \hline
        CZ$^\mathrm{f}_{ij}$ & $\gamma_i\mathrm{Z}_j^\mathrm{f}$ & $\tilde\gamma_i\mathrm{Z}_j^\mathrm{f}$ & $\gamma_j\mathrm{Z}_i^\mathrm{f}$ & $\tilde\gamma_j\mathrm{Z}_i^\mathrm{f}$ & & \\
        \hline
        CZ$^\mathrm{qf}_{ij}$ & & & $\mathrm{Z}_i\gamma_j$ & $\mathrm{Z}_i\tilde\gamma_j$ & $\mathrm{Z}_i$  & $\mathrm{Z}_j^\mathrm{f} \mathrm{X}_i$ \\
        \hline
    \end{tabular}
    \caption{\textbf{Transformation of Majorana and qubit operators under our gate set.} We show $U^\dagger A U$, where $U$ is specified in the first column and $A$ in the first row. The indices are defined such that braiding is given by $\mathrm{BRAID}_{ij}=\exp\left(-\frac{\pi}{4} \tilde \gamma_i \gamma_j\right)$ and qubit-fermion controlled-Z as $\mathrm{CZ}^\mathrm{qf}_{ij}=\exp\left(-i\pi n^\mathrm{q}_i n_j \right)$. An empty field means that the operator is invariant. Note that $\mathrm{Z}_i^\mathrm{f}=-i\tilde \gamma_i \gamma_i$.}
    \label{tab:gates}
\end{table}

\subsubsection*{Logical \texorpdfstring{$\mathrm{CZ}^\mathrm{f}$}{CZf} gate} 
Writing the logical $\mathrm{CZ}^\mathrm{f}$ gate between two codeblocks $a$ and $b$ in terms of logical Majoranas, we get
\begin{equation}
    \exp\left(-i\pi  n^a  n^b \right) = \frac{1}{2}\left(1+i\tilde \gamma^a_N \gamma^a_1 +i\tilde \gamma^b_N \gamma^b_1 - \tilde \gamma^a_N \gamma^a_1\tilde \gamma^b_N \gamma^b_1 \right).
\end{equation}
We construct this gate from a physical CZ$^\mathrm{f}$ gate between the two edge sites, which in terms of the Majoranas is written as 
\begin{align}
    &\exp\left(-i\pi  n^a_1  n^b_1 \right)\notag\\&= \frac{1}{2}\left(1+i\tilde \gamma^a_1 \gamma^a_1 +i\tilde \gamma^b_1 \gamma^b_1  - \tilde \gamma^a_1 \gamma^a_1 \tilde \gamma^b_1 \gamma^b_1\right),
\end{align}
where $n^a_1$ is the number operator acting on the first site of codeblock $a$.
While the general form is already correct, we need to replace $\tilde \gamma^b_1\rightarrow \tilde \gamma^b_N$ and $ \tilde \gamma^a_1 \rightarrow \tilde\gamma^a_N$. Such a ``replacement operation'' is performed by conjugating with braidings (see  Table~\ref{tab:gates}):
\begin{align}
    &\exp\left(-i\pi  n^a  n^b \right)\notag\\ &= e^{\frac{\pi}{4} \tilde\gamma_N^b \tilde\gamma_1^b} e^{\frac{\pi}{4}   \tilde\gamma^a_N \tilde\gamma^a_1} e^{-i\pi  n^a_N n^b_1} e^{-\frac{\pi}{4}   \tilde\gamma^a_N\tilde\gamma^a_1}e^{-\frac{\pi}{4} \tilde\gamma_N^b  \tilde\gamma_1^b} .
\end{align}
Using Eq.~\eqref{eq:arbbraids} to express this sequence in terms of BRAID and \Sf gates, we find the circuit shown in Fig.~\ref{fig:2}d. Note that this compilation also works in the same way for an arbitrary-angle gate $\exp(-i\theta n^a n^b)$ by replacing the physical \czf with $e^{-i\theta  n^a_N n^b_1}$. 

\subsubsection*{Logical braiding gate gate.} A logical braiding gate is given in terms of the Majorana fermions as
\begin{equation}
    \mathrm{BRAID}_{ab}=\exp\left(-\frac{\pi}{4}\tilde \gamma^a_N \gamma^b_1 \right).
\end{equation}
This is in fact a physical braiding gate between the edges of two codeblocks, as depicted in Fig.~\ref{fig:2}d.

\subsubsection*{Logical qubit-fermion gate.}
Similarly, we decompose the logical qubit-fermion gate as
\begin{equation}
    \exp\left(-i \pi n^\mathrm{q} n^a \right)=e^{\frac{\pi}{4} \gamma_1^a \gamma_N^a} e^{-i\pi n^\mathrm{q} n_N^a}e^{-\frac{\pi}{4} \gamma_1^a \gamma_N^a}.
\end{equation}
Again using Eq.~\eqref{eq:arbbraids}, we find the circuit shown in Fig.~\ref{fig:2}d.

\subsubsection*{Logical phase gate.} A logical fermion phase gate $\exp\left(i\theta  n^a\right)$ is written in terms of the Majorana fermions as
\begin{equation}
    \exp\left(i \frac{\theta}{2}  (1+i\tilde \gamma^a_N \gamma^a_1) \right).
\end{equation}
Therefore, this gate is, up to a global phase, equivalent to an arbitrary-angle braiding between the two edge Majoranas. In particular, the \Sf gate is equivalent to \braid and the \Tf gate is a $\pi/8$ braiding.

\subsubsection*{Syndrome measurement.} For the measurement of the syndromes $i\tilde \gamma_i \gamma_{i+1}$, we study the following transformation of the final $Z$ measurement of the ancilla qubits: $U^\dagger Z_\mathrm{anc} U$ where we use $U=\mathrm{H}_\mathrm{anc} \mathrm{CZ}^\mathrm{qf}_{\mathrm{anc},i+1} \mathrm{BRAID}_{i+1,i}\mathrm{S}^\mathrm{f}_i \mathrm{H}_\mathrm{anc}$. Using the transformation rules in Table \ref{tab:gates}, we find $U^\dagger Z_\mathrm{anc} U=Z_\mathrm{anc}i\tilde \gamma_i \gamma_{i+1}$, and, therefore, initializing the ancilla in $\ket{0}$, a measurement of the ancilla yields a measurement of $i\tilde \gamma_i \gamma_{i+1}$. The additional gates in Fig.~\ref{fig:2}b compared to the circuit defined here are needed to not disturb the logical information. This enables true mid-circuit error correction.

\subsection*{Color code logical operations}
In this subsection, we show that the circuits in Fig.~\ref{fig:3} for \czf, \braid, and \czqf are indeed logical gates. 

We show that stabilizers transform to products of stabilizers and that the gates transform the logical Majorana operators in exactly the same way as the physical gates transform physical Majorana operators (c.f.~Table \ref{tab:gates}). Therefore, we study the unitary transformation $U^\dagger\gamma^\mathrm{L} U$ and $U^\dagger\tilde \gamma^\mathrm{L} U$ of the logical Majorana operators under the circuits shown in Fig.~\ref{fig:3}. All of the gates are transversal because each physical gate acts only on a single physical fermion in each codeblock, similarly to the qubit-only architecture in Ref.~\cite{Bluvstein2023}.

Similar to the notation in the previous Methods subsection on repetition code logical operations, we attach a superscript to physical Majorana operators to indicate the codeblock and let the physical site index $i$ only run from $1$ to the number of physical sites within the codeblock. For example, $\gamma^a_i$ indicates physical lattice site $i$ in codeblock $a$, where $i\in \lbrace 1,7\rbrace$ for the seven-site triangular color code. The below calculations hold for all triangular color codes.

\subsubsection*{\texorpdfstring{$\mathrm{BRAID}$\,}{BRAID }gate} 

The gate
\begin{equation}
\mathrm{BRAID}_{ab}=\prod_i \exp\left(-\frac{\pi}{4} \tilde \gamma_i^a \gamma_i^b\right)
\end{equation}
is a logical braiding gate. To see this, first note that the individual braiding gates commute with Majorana operators on sites $j\neq i$ because they are quadratic. 
Therefore, this gate transforms the logical operators $\tilde \gamma^a=\prod_i \tilde \gamma^a_i$ and $\gamma^b=\prod_i \gamma^b_i$ as
\begin{align}
&\mathrm{BRAID}_{ab}^\dagger \tilde \gamma^\mathrm{L}_b  \mathrm{BRAID}_{ab} \notag\\&= \prod_i \exp\left(\frac{\pi}{4} \tilde \gamma_i^a \gamma_i^b\right) \gamma^b_i \exp\left(-\frac{\pi}{4} \tilde \gamma_i^a \gamma_i^b\right)\notag\\ &= \prod_i  \tilde\gamma^a_i = \tilde\gamma^\mathrm{L}_a,\\&\mathrm{BRAID}_{ab}^\dagger \tilde \gamma^\mathrm{L}_a  \mathrm{BRAID}_{ab} \notag\\&= \prod_i \exp\left(\frac{\pi}{4} \tilde \gamma_i^a \gamma_i^b\right) \tilde \gamma^a_i \exp\left(-\frac{\pi}{4} \tilde \gamma_i^a \gamma_i^b\right)\notag\\ &= \prod_i \left(-\gamma^b_i\right)=-\gamma^\mathrm{L}_b, 
\end{align}
where in the very last step we used the fact that the triangular color codes have an odd number of physical fermions.

Stabilizers transform to stabilizers because they are products of an even number of Majorana fermions consisting of only either $\gamma$ or $\tilde \gamma$, but not both. Braiding will transform them to products of the other type residing on the same lattice sites of the other codeblock. Because of the self-duality (in the sense of the symmetry $\gamma \leftrightarrow \tilde \gamma$ of the stabilizers) of our code, the resulting product is also a stabilizer. Because stabilizers consist of an even product of Majorana operators, there is also no minus sign resulting from the transformation.

\subsubsection*{\texorpdfstring{\Sf\,}{Sf} gate} The transversal gate
\begin{equation}
\mathrm{S}^\mathrm{f}_{a}=\prod_i \mathrm{S}^\mathrm{f}_{i}
\end{equation}
is shown to be a logical phase gate \Sf\,in the same way as the BRAID gate.

\begin{figure*}
    \centering
\includegraphics[width=\textwidth]{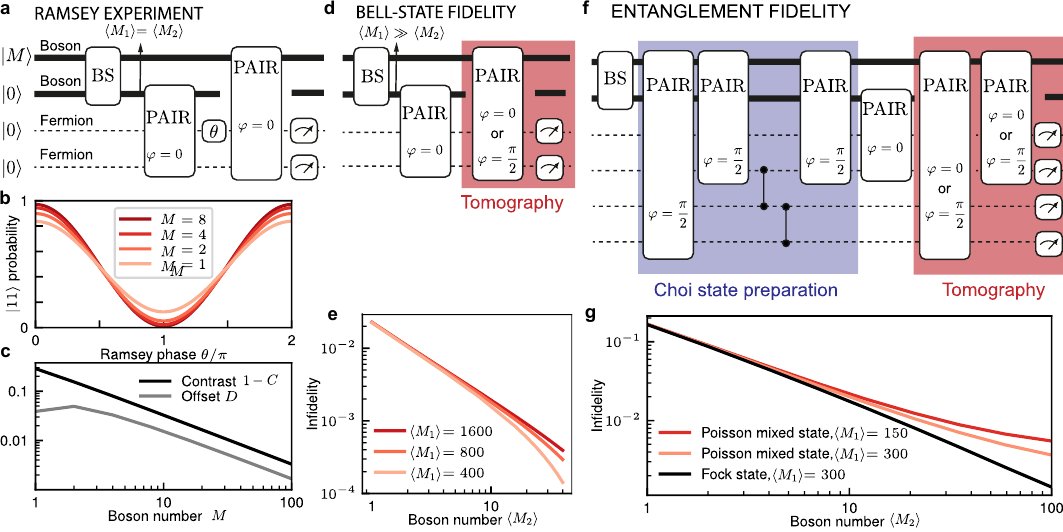}    \caption{\textbf{Benchmarking the pairing gate.} Pairing gate definition in Eq.~\eqref{eq:pairmoregeneral}, with the occupation number entering its phase given by $\braket{M_i}$ of the mode $i$ it acts on. \textbf{a}, Circuit for implementing the Ramsey circuit, reproduced from Fig.~\ref{fig:5}c of the main text. The beamsplitter phase is $\pi/4$. \textbf{b}, Ramsey fringes, reproduced from Fig.~\ref{fig:5}d of the main text. \textbf{c}, Contrast $C$ and offset $D$ obtained from fitting $\frac{1}{2}+\frac{1}{2}\left(C\cos(\theta)-D\right)$ to the Ramsey fringes in subfigure b. \textbf{d}, Circuit for performing tomography within the $\ket{00}$ and $\ket{11}$ manifold. The beamsplitter phase is $\arccos\left(\sqrt{\braket{M_2}/(\braket{M_1}+\braket{M_2})}\right)$, chosen such that the average occupation number of the first, second bosonic mode (counted from the top) is $\braket{M_1}$, $\braket{M_2}$ after the beamsplitter, respectively. \textbf{e}, Infidelity $1-\braket{\psi|\rho|\psi}$, where $\ket{\psi}=(\ket{00}+i\ket{11})/\sqrt{2}$ and $\rho$ is the density matrix reconstructed from the numerical tomography. Initial state is the same is in a. \textbf{f}, Circuit for numerically measuring the average gate fidelity for all possible input states of the form $\alpha\ket{00}+\beta\ket{11}$. The Choi state is defined as $\ket{\psi_\mathrm{Choi}}=(\ket{0000}+\ket{1111})\sqrt{2}$. The tomography is performed within the subspace spanned by the states $\ket{0000}$, $\ket{0011}$, $\ket{1100}$, $\ket{1111}$. \textbf{g}, Infidelity $1-\bra{\psi_\mathrm{Choi}}\rho\ket{\psi_\mathrm{Choi}}$ for either a Fock state input $\ket{M}$ on mode 1 (black line) or a mixed state on mode 1 with a Poissonian number distribution with average boson number $\braket{M_1}+\braket{M_2}$ (colored lines). The black line is shown in Fig.~\ref{fig:5}e of the main text.}
\label{fig:ED}
\end{figure*}

\subsubsection*{\texorpdfstring{\czf\,}{CZf }gate} The \czf gate is given by
\begin{equation}\mathrm{CZ}^\mathrm{f}_{ab}=\prod_i \exp\left(-i\pi n^a_i n^b_i\right).
\end{equation}
To check the transformation of the logical operators, consider

\begin{align}
   \left( \mathrm{CZ}^\mathrm{f}_{ab} \right)^\dagger\gamma^a \mathrm{CZ}^\mathrm{f}_{ab} &= \prod_i \exp\left(i\pi n^a_i n^b_i\right) \gamma^a_i\exp\left(-i\pi n^a_i n^b_i\right) \\&= \prod_i (-i\gamma^a_i\tilde\gamma^b_i\gamma^b_i)\\& = -i \gamma^\mathrm{L}_a \tilde\gamma^\mathrm{L}_b\gamma^\mathrm{L}_b.
\end{align}
The last step is non-trivial: the global sign  is a combination of the overall $(-i)^n$ and a minus sign resulting from reordering the physical Majorana fermions. The triangular Honeycomb color codes have $n=(3d^2+1)/4$ physical fermions for odd $d$~\cite{Litinski2017}. For convenience, we define a $\tilde d$ as $d=2\tilde d+1$, and write $n=3\tilde d^2+3\tilde d+1$. One can show that $\prod_{i=1}^n(\gamma^a_i\tilde\gamma^b_i\gamma^b_i)=(-1)^{\lfloor n/2 \rfloor} \prod_{i=1}^n\gamma^a_i\prod_{i=1}^n\tilde\gamma^b_i\prod_{i=1}^n\gamma^b_i$. Therefore, the overall prefactor is given by $(-i)^{3\tilde d^2+3\tilde d+1} (-1)^{\lfloor (3\tilde d^2+3\tilde d+1)/2 \rfloor}=-i$ for all integer $\tilde d$.

Similarly, stabilizers transform to products of stabilizers after the transformation. To see the global sign, note that, in the triangular color codes, stabilizers are products of either four or six Majorana operators. For four operators, $(-i)^4=1$ and the sign from reordering is $(-1)^{\lfloor 4/2 \rfloor}=1$. For six operators, $(-i)^6=-1$ and the sign from reordering is $(-1)^{\lfloor 6/2 \rfloor}=-1$, again yielding an overall $+1$.

\subsubsection*{\texorpdfstring{\czqf\,}{CZqf }gate}  Our qubit-fermion gate \czqf is implemented as \begin{equation}
\mathrm{CZ}^\mathrm{qf}_{ab}=\prod_i \exp\left(-\pi n^{\mathrm{q},a}_i n^b_j\right).
\end{equation}
This is shown in the same manner as the previous gates. In particular, the fact that $\prod_{i=1}^n(\tilde\gamma^a_i\gamma^a_i)=(-1)^{\lfloor n/2 \rfloor} \prod_{i=1}^n\tilde\gamma^a_i\prod_{i=1}^n\gamma^a_i$  leads to the correct minus signs. Stabilizers map onto products of fermion and qubit stabilizers under the action of this gate. Therefore, it is a valid logical qubit-fermion gate.

\section*{Experimental considerations}

In this section, we present some details of the neutral-atom scheme, in particular for implementing pairing. 

\subsection*{Non-local pairing from on-site pairing}

We demonstrate that the pulse sequence shown in Fig.~\ref{fig:5}a of the main text indeed implements the target inter-site braiding gate
\begin{equation}
    \mathrm{PAIR}_{ij}=\exp\left(i\frac{\pi}{4\sqrt{M}}(e^{i\varphi}b_k p_i^\dagger p_j^\dagger + h.c.)\right), \label{eq:pairmoregeneral}
\end{equation}
where we generalized the definition in the main text to include a phase $\varphi$. The pulse sequence is arranged such that the dissociation gate
\begin{equation}
    U_\mathrm{dis}=\exp\left(-i\frac{\pi}{4\sqrt{M}}(e^{i\varphi}b_i p_{i\uparrow}^\dagger p_{i\downarrow}^\dagger + h.c.)\right)
\end{equation}
is conjugated with a series of unitaries. We therefore have to show that these unitaries transform bosonic and fermionic creation/annihilation operators in such a way that $U_\mathrm{dis}$ is transformed into $\mathrm{PAIR}_{ij}$. First, the beamsplitter gate
\begin{equation}
\mathrm{BS}=\exp\left(\frac{\pi}{4}(b_i^\dagger b_k - h.c.)\right)
\end{equation}
acts as $(\mathrm{BS}^\dagger)^2 b_i^\dagger \mathrm{BS}^2=b_k^\dagger$. Hence,  $(\mathrm{BS}^\dagger)^2 U_\mathrm{dis} \mathrm{BS}^2=\exp\left(-i\frac{\pi}{4\sqrt{M}}(b_k p_{i\uparrow}^\dagger p_{i\downarrow}^\dagger + h.c.)\right)$. Similarly, 
\begin{equation}
    U_\mathrm{t\downarrow}=\exp\left(-i\frac{\pi}{2}(p_{i\downarrow}^\dagger p_{j\downarrow} + h.c.)\right)
\end{equation}
acts as $U_\mathrm{t\downarrow}^\dagger p_{i\downarrow}^\dagger U_\mathrm{t\downarrow}=i p_{j\downarrow}^\dagger$ and
\begin{equation}
    U_\mathrm{\uparrow\downarrow}=\exp\left(-i\frac{\pi}{2}(p_{i\uparrow}^\dagger p_{i\downarrow} + h.c.)\right)
\end{equation}
acts as $U_\mathrm{\uparrow\downarrow}^\dagger p_{i\uparrow}^\dagger U_\mathrm{\uparrow\downarrow}=ip_{i\downarrow}^\dagger$. Therefore, 
\begin{equation}
\mathrm{PAIR}_{ij}=U^\dagger_\mathrm{\uparrow\downarrow}U^\dagger_\mathrm{t\downarrow}(\mathrm{BS}^\dagger)^2U_\mathrm{dis}\mathrm{BS}^2U_\mathrm{t\downarrow}U_\mathrm{\uparrow\downarrow},
\end{equation}
which is the sought-after pulse sequence when identifying spinless fermions with fermions in the $\downarrow$ state. The Hermitian conjugate of each gate is performed by conjugating with phase gates, e.g.~$U^\dagger_\mathrm{\uparrow\downarrow}=Z^\mathrm{f}_\downarrow U _\mathrm{\uparrow\downarrow} Z^\mathrm{f}_\downarrow$.

\subsection*{Benchmarking the pairing gate}

Here we show that Eq.~\eqref{eq:pairmoregeneral} indeed effectively implements a pairing gate as defined in Eq.~\eqref{eq:pairtarget} when the molecule/boson mode on which it acts has large occupation. To do so, we discuss three different standard benchmark experiments and study them numerically: Ramsey spectroscopy, state tomography of a Bell state, and process/gate tomography based on measuring the entanglement fidelity.

\subsubsection*{Ramsey experiment} First, we want to test whether there is effective coherence created after applying the pairing gate. As for qubits, this is done by performing a Ramsey-type experiment, i.e. applying two $\pi/2$ pairing pulses, with a phase gate in-between. At the same time as testing whether this is the case, we also test whether two separate molecule/boson modes can be used for the two pulses. It is clear that these gates need to be effectively phase coherent. We achieve this by creating two modes with Poissonian number distributions, mimicking a coherent state, by applying a 50/50 beamsplitter on the state $\ket{M}$, see Fig.~\ref{fig:ED}a (also shown in Fig.~\ref{fig:5}c of the main text). Being able to perform both gates using separate bosonic modes is important---this tests whether ultimately we can parallelize pairing operations. In Fig.~\ref{fig:ED}b (also shown in Fig.~\ref{fig:5}d of the main text), we show the result of the Ramsey sequence. We indeed find a Ramsey fringe, verifying that effective coherence has been created in the fermions. The contrast increases with $M$ as shown in Fig.~\ref{fig:ED}c, and a slight offset of the oscillations is visible. More quantitatively, we fit an oscillatory function to the fringe (see figure caption for its definition), finding that the contrast approaches unity as $\propto 1/M$. At the same time, the offset decreases with the same functional dependence. Hence, in the limit $M\rightarrow \infty$, perfect Ramsey oscillations are approached.

 \subsubsection*{Bell-state fidelity} In order to probe that the state after the pairing gate is not only an effective coherent superposition, but exactly the state that we aim to prepare, in Fig.~\ref{fig:ED}d we construct a circuit effectively performing Bell-state tomography. To do so, we perform three separate circuits: we always first apply a pairing gate to effectively prepare the Bell state $\ket{00}+i\ket{11}$ and then perform either a pairing gate with $\varphi=0$ or $\varphi=\pi/2$ or no pairing gate at all. From these three measurements, we reconstruct the density matrix in the $\ket{00}$ and $\ket{11}$ manifold. In order to separate the boson-number dependence of the preparation step from the tomography step, we choose the beamsplitter phase such that the boson number $M_1$ of the mode we use to perform the second gate is much larger than the boson number $M_2$ of the mode we use to perform the first gate. Having reconstructed the density matrix, we show the infidelity of the Bell-state preparation as a function of $M_2$, for different values of $M_1$. We find in Fig.~Fig.~\ref{fig:ED}e that the curves collapse as $M_1\rightarrow \infty$, indicating again a power-law dependence of $\propto 1/M$. Crucially, already for $M_2\approx 20$, we find a Bell-state fidelity of approx. $99.9\,\%$.

\subsubsection*{Gate fidelity} Finally, we would like to benchmark our gate for all possible initial states. To do so, we perform gate/process tomography by measuring the entanglement fidelity $F_\mathrm{e}$, relying on the relation
\begin{equation}
    F_\mathrm{avg}=1-\frac{dF_\mathrm{e}+1}{d+1},
\end{equation}
where $d$ is the Hilbert-space dimension and
\begin{equation}
    F_\mathrm{avg}=\int \mathrm{d}\psi \braket{\psi|U^\dagger \mathcal{E}(\ket{\psi}\bra{\psi}) U|\psi}
\end{equation}
is the state-averaged fidelity between the target unitary $U$ and the realized channel $\mathcal{E}$. The entanglement fidelity is measured by first preparing the Choi state $\ket{\psi_\mathrm{Choi}}=\frac{1}{\sqrt{d}}\sum_{i=1}^d \ket{i}\ket{i}$, constructed using an ancillary system of the same size as the target system, applying the target channel $\mathcal{E}\otimes \mathds{1}$, and finally performing state tomography on the resulting state. In our case, $\mathcal{E}$ is the pairing gate and hence all gates stay in the even parity sector such that $d=2$, $\ket{\psi_\mathrm{Choi}}=\frac{1}{\sqrt{2}}\left(\ket{00}\ket{00}+\ket{11}\ket{11} \right)$, and the tomography is only performed in that sector.

We show the entanglement fidelity in Fig.~\ref{fig:ED}g, black line (same as Fig.~\ref{fig:5}e), as a function of the boson number $M_2$ in the second mode, where the result is converged with respect to $M_1$. We find that $>99.8\,\%$ gate fidelity is achieved for $M_2=100$. 

Finally, in experiment, it is much easier to prepare a mixed state with Poissonian number distribution with average number $\bar M=M_1+M_2$ rather than a Fock state, for example by preparing a Bose-Einstein condensate~\cite{Jochim2003} in that tweezer ($\bar M$ is the number of condensed molecules in the ground state of the tweezer). In that case, there will be a shot-to-shot miscalibration of the angle of the pairing gate, leading to a further source of infidelity. To check the robustness of our gate with respect to this effect, we replace the initial state $\ket{M}$ of the molecules with a state $\rho_{\bar M}\otimes \ket{0}\bra{0}$, where $\rho_{\bar M}=\sum_n p_{\bar M}(n) \ket{n}\bra{n}$ and $p_{\bar M}(n)=\frac{\bar M^n e^{-\bar M}}{n!}$ is a Poisson distribution with mean $\bar M$. The phase of the pairing gate is fixed to $\pi/\left(4 \sqrt{\bar M}\right)$. We show the resulting average infidelity as red lines in Fig.~\ref{fig:ED}g. We find that, for small $M_2$, the result is indistinguishable from the Fock-state result. For larger $M_2$, deviations occur. These deviations become smaller as $M_1$ increases. Hence, we can expect that as $\bar M$ and $M_2$ increase, arbitrarily high fidelities can be achieved even when the initial molecule state is a mixed state with a Poissonian number distribution.

\subsection*{Experimental time scales}

We discuss more detailed considerations for the parameter regimes that need to be fulfilled to implement the pairing scheme.

For spin-selective tunneling $U_{t\downarrow}$, a magnetic field gradient needs to be applied to suppress tunneling of the $\uparrow$ spin species. Similarly, to guarantee that no fermionic atoms tunnel from the target tweezer into the BEC tweezer during the application of the beamsplitter ($\mathrm{BS}$), we apply a magnetic field gradient and detune the two trap frequencies such that the detuning induced by the magnetic field gradient is compensated by the tweezer detuning for the molecules. Because the molecules carry twice the magnetic moment of atoms, the atom tunneling is still suppressed by the residual gradient. For dissociation, the rf frequency $\nu_\mathrm{RF}$ is chosen to bridge the bare atom energy $\nu_\uparrow-\nu_\downarrow$ as well as the binding energy $E_\mathrm{B}$ of the molecule and the trap depth $\omega_\mathrm{tweezer}$ of the target tweezer (which again enters due to the differing magnetic moments of atoms and molecule), $\nu_\mathrm{RF}=\nu_\uparrow-\nu_\downarrow +E_\mathrm{B} + \omega_\mathrm{tweezer}$. This can be easily fulfilled: typical trap frequencies are in the tens of kHz, and binding energies in ${}^6$Li$_2$ are in the hundreds of kHz, for ${}^{84}$Sr$_2$ even in the hundreds of MHz~\cite{Regal2003,Chin2004,Chin2005}. Binding energies may be tuned by a Feshbach resonance, however this needs to be counterbalanced with the interactions between bosons and fermions as well as between the $\downarrow$ and $\uparrow$ fermions, which need to be small for the duration of the gate. The coupling $g$ also needs to be much smaller than the tweezer trap frequency in order to guarantee that the atom pair is created in the motional ground state of the tweezer. 

\subsection*{Atom movement}
\begin{figure}
    \centering
    \includegraphics[width=1\linewidth]{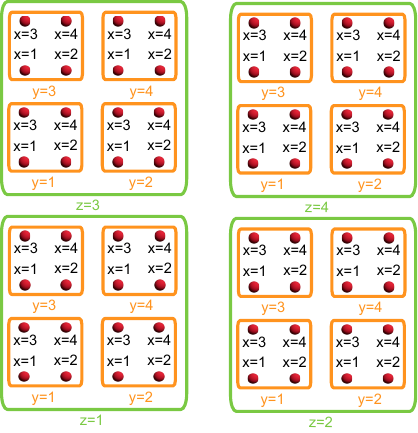}
    \caption{\textbf{Spatial arrangement of sites for 3D FFT.} The index of a site is written with respect to the $(x,y,z)$ coordinates. The Fourier transforms can be applied in parallel on each orange block.}
    \label{fig:fftarrangement}
\end{figure}
In this section, we show that the movement-time-scaling with $N$ can be neglected under reasonable experimental assumptions because of the constant (with $N$) overhead from having to perform $\mathcal{O}(1)$ quantum error correction cycles after every layer of gates and the spatial structure of the high-dimensional Fourier transform circuits. This is important in order to translate our exponential circuit depth advantage into a corresponding improvement in circuit runtime when modeling the circuit runtime as depth multiplied by time per layer of gates.

We consider atoms arranged in a two-dimensional grid. In that case, the scaling of the maximum movement time is rather slow with $N$: due to the fact that the tweezers are constantly accelerated/decelerated during movement, the time scales as $t=2\sqrt{s/a}$, where $a$ is the acceleration and $s$ is the distance. The maximum distance is given by $\sqrt{2N}a_\mathrm{lat}$, where $a_\mathrm{lat}$ is the lattice spacing. Therefore, the maximum movement time only scales as $t\propto N^{1/4}$. Assuming a lattice spacing of approx. $a_\mathrm{lat}= 5\mu$m (this is conservative - with an optical lattice, $a_\mathrm{lat}= 0.5\mu$m could be achieved, although more complicated movement schemes would have to be employed in that case), the longest movement time for $N=1$ million (presumably close to the number of atoms controllable in a single array) and assuming $a\approx 10^4$ m/s$^2$ (see e.g. ~\cite{Bluvstein2022}) is approximately $2$ms. This movement time needs to be compared to the constant overhead from quantum error correction. One quantum error correction cycle takes approximately $1$ms~\cite{Bluvstein2023} and while recently, it was shown that one such cycle per layer of gates suffices in the surface code~\cite{zhou2024}, it is unclear how this generalizes to color codes, which are usually more challenging to decode and therefore may require multiple rounds of QEC per gate layer. Therefore, even the longest movement time in a large array of $1$ million sites is comparable to the constant overhead from QEC.

Moreover, long-distance moves are rarely needed for simulating $3$D materials. Long-distance moves only appear in the 3D FFFT and 3D quantum FFT. Both are implemented as three independent 1D Fourier transforms along the individual dimensions. To apply these transforms move-efficiently, we first arrange the sites as shown in Fig.~\ref{fig:fftarrangement}. The transform along the x coordinate can then be applied in parallel on all (y,z) blocks. To do so, only moves of distance $\propto N^{1/6}$ are required. To apply the y transform, we first re-arrange the sites within each block of fixed $z$ coordinate such that now y and x coordinate are exchanged in Fig.~\ref{fig:fftarrangement}. While this requires moves of distance $\propto N^{1/4}$, the FFT (which is a much deeper circuit) then again only requires moves of distance $\propto N^{1/6}$. Finally, to apply the FFT along the z coordinate, we exchange z and y coordinates by moves. This is the only step requiring moves of distance $\propto N^{1/2}$.

Another hardware restriction comes from routing the qubits into arbitrary permutations. It was recently shown that using a straightforward upgrade to current atom-movement implementations (adding a second SLM) enables a routing algorithm which succeeds in $\Theta(\log_2(N))$ steps~\cite{constantinides2024}, therefore providing negligible overhead.

Putting everything together, the total time of the algorithm is given by an expression of the form $\log_2(N)(c_\mathrm{qec}+c_1N^{1/12}\log_2(N)+c_2 N^{1/8}\log_2(N)+c_3 N^{1/4}\log_2(N))$, where the overall $\log_2(N)$ comes from the algorithmic depth, the $\log_2(N)$ inside the bracket from routing. Because the circuits to apply the Fourier transforms are much deeper than the movement circuits between the three Fourier transforms, $c_\mathrm{qec}>c_1\gg c_2,c_3$ holds. As discussed above, using practically relevant $N$, $c_\mathrm{qec}$ will not be overwhelmed by terms inside the parentheses that scale with $N$ and to a good approximation, the leading time scaling is logarithmic up to the maximum realistic $N$.

\subsection*{Implementation in dots and donors in silicon}

We discuss a solid-state implementation in dots or donors, where spin-polarized electrons encode fermions. 

Qubits with long coherence times can be encoded in nuclear spins, in which case a $\mathrm{CZ}^{\mathrm{qf}}$ gate is implemented using the hyperfine interaction~\cite{Rad2024}. Alternatively, a qubit can be encoded in the electron spin. Both qubit-qubit and qubit-fermion entangling gates are then implemented by Heisenberg exchange~\cite{Petta2005,Burkard2023}. Single-fermion gates may be realized using electric-field gates, and the pairing gate using the proximity effect to a superconductor~\cite{Stanescu2010}, where Cooper pairs take the role of molecules. 

\section*{Algorithms for quantum simulation}

\subsection*{Time evolution under the materials interaction term in logarithmic depth}
In this section, we discuss how ancillas can be used to reduce the circuit depth of applying the interaction part of the materials Hamiltonian to depth $\mathcal{O}(\mathrm{log}(N))$. In other words, we aim to implement the time-evolution operator
\begin{equation}
    \hat U = \mathrm{exp}\left(-it \sum_{i<j} J_{ij} n_i n_j \right).
\end{equation}
Without ancillas, currently-known schemes achieve $\mathcal{O}(N)$ depth in the worst case due to the fact that only $\mathcal{O}(N)$ gates can be applied in parallel (assuming access to only single and two-qubit gates). 

We first discuss the fan-out scheme introduced in Refs.~\cite{moore2001,hoyer2005}, which requires $\mathcal{O}(N^2)$ ancillas. It is based on encoding every site into an $(N-1)$-site bit-flip repetition code. In other words, we map the state of a single qubit to
\begin{equation}
	\alpha \ket{0} + \beta \ket{1} \rightarrow \alpha \ket{0\cdots 0} + \beta \ket{1 \cdots 1}. 
\end{equation}
The logical $\bar Z$ operator of logical qubit $i$ is degenerate, in the sense that acting with $\hat Z$ on any of its $N-1$ physical qubits implements $\bar Z$. Because in $\hat U$, every site interacts with $N-1$ other sites, this means that we can apply all $ZZ(\phi_\mathrm{ij})$ gates in parallel in depth one. Afterwards, we decode the repetition codes and discard the ancillas. The repetition code state preparation requires a fan-out gate, which can be decomposed into a measurement-free circuit of depth $\mathcal{O}(\log(N))$, or a circuit with depth $O(1)$ using measurements~\cite{pham2013}. Alternatively, a fan-out gate can be realized with continuous Hamiltonian time evolution in a time constant in the number of qubits $N$ with power-law decaying interactions~\cite{guo2022} or even in a time of $O(\log^2 N/N)$ with all-to-all interactions~\cite{yin2025}. In any case, this scheme enables applying $\hat U$ in total depth $\mathcal{O}(\log(N))$ while using $N(N-2)$ ancillas.

We now discuss a scheme, in the context of quantum computing first presented in section V.A of Ref.~\cite{low2019a}, to reduce the number of ancillas to $\mathcal{O}(N\log(N))$ for time evolution under the interaction part of the materials Hamiltonian by using its translational invariance.  A key observation is that, for the materials Hamiltonian, the interaction matrix is translationally invariant, i.e.
\begin{equation}
    J_{ij}=J_{|i-j|}
\end{equation}
in one spatial dimension. This means that there is redundant information in the matrix which we can remove by Fourier transformation:
    \begin{equation}
    \hat U = \mathrm{exp}\left(-it \sum_{i<j} J_{ij} n_i n_j \right) = \mathrm{exp}\left(-it \sum_{k} J_k |n_k|^2 \right),
\end{equation}
where $\ n_k=\sum_j e^{-ikj} n_j$ and $J_k=\sum_r e^{-ikr} J_r=1/k^2$ for materials. Note that, because $J_r$ is real and symmetric, $J_k$ is real.
In Fourier space, there are only $N$ terms in the Hamiltonian and therefore, it should be possible to apply $\hat U$ in depth $\mathcal{O}(1)$. However, to actually perform the rotations in this way, we need to have access to the Fourier-transformed occupation numbers $n_k$, which are complex numbers.

The scheme in Ref.~\cite{low2019a} enables this by using a quantum circuit of the classical fast Fourier transform, discussed in more detail in Ref.~\cite{asaka2020}. The intuitive reason why $\mathcal{O}(N \mathrm{log}(N))$ ancillas suffice is because we can encode $n_k$ for a single $k$ with $\mathcal{O}(\log(N))$ ancillas using a fixed-precision binary register. The quantum version of the classical FFT can be applied in 
$\mathcal{O}(\log(N))$ depth for $N=2^n$ and integer $n$. 

For completeness, we describe how this idea works~\cite{low2019a}. We describe its action on a single occupation number string, e.g. $\ket{0110\cdots 110}$, but the algorithm is fully reversible and acts on all strings in the many-body superposition in parallel:\\
1) Encode binary occupation numbers $\hat n_i$ into a fixed-precision complex register \\
2) Apply the quantum circuit of the classical FFT on these registers, after which $n_k$ is stored in the registers\\
3) Multiply each of the registers with themselves, yielding $|n_k|^2$.\\
4) Apply phase rotations in parallel on all registers, applying $\exp\left(itJ_k |n_k|^2\right)$.\\
5) Apply the inverse QFFT.\\
6) Unencode the registers back into the sites.

All of these operations require at most logarithmic depth, and, therefore, this circuit evolves under the interaction term of the materials Hamiltonian in total depth $\mathcal{O}(\log(N))$, using $\mathcal{O}(N\log(N))$ ancillas.

\subsection*{Speeding up chemistry in the molecular basis}

In this section, we show how to implement a single Trotter step under the quantum chemistry Hamiltonian in the molecular basis in  fermionic quantum computers in linear depth for large molecules, compared to quadratic depth in current qubit methods.

We first reformulate the quantum chemistry Hamiltonian $H_\mathrm{qchem}$ using the double factorization approach~\cite{motta2021} which replaces the four-fermion term with a density-density interaction and which is currently the most efficient formulation for simulating molecules~\cite{gunther2025},
\begin{align}
&\sum_{i<j<k<l=1}^NV_{ijkl} c^{\dagger}_ic^{\dagger}_j c_kc_l\\ &= \sum_{i<j=1}^N S_{ij} c^{\dagger}_i c_j + \frac{1}{2}\sum_{\ell=1}^L \left(U^{(\ell)}\right)^{\dagger} \left(\sum_{p=1}^{\rho_{\ell}} \lambda^{(\ell)}_{p} n_{p} \right)^2{U}^{(\ell)}.
\end{align}
The $U^{(\ell)}$ are fermionic basis rotations, implemented via a product of Givens rotations. $\rho_{l}\leq N$ is a truncation parameter. The parameters $S_{ij}$, $\lambda^{(\ell)}_{p}$ can be efficiently computed classically from the $V_{ijkl}$, as can be which Givens rotations need to be performed within the $U^{(\ell)}$.  

A single first-order Trotter step is then given by
\begin{equation}
e^{i \Delta t H_\mathrm{qchem}}=e^{i \Delta t(h+S)} U^{(1)} \prod_{\ell=1}^L e^{i \Delta t V^{(\ell)}} \tilde{U}^{(\ell)}+\mathcal{O}(\Delta t)^2,
\end{equation}
where $h$ is the quadratic part of the Hamiltonian,  $\tilde{U}^{(\ell)}=U^{\dagger(\ell)} U^{(\ell+1)}$, and $V^{(\ell)}= \sum_{i j=1}^{\rho_{\ell}} \frac{\lambda_i^{(\ell)} \lambda_j^{(\ell)}}{2} n_i^{(\ell)} n_j^{(\ell)}$.

For small $N$, $\rho_\ell=\mathcal{O}(N)$ and $L=\mathcal{O}(N)$~\cite{motta2021}. Therefore, $V^\ell$ requires circuit depth $N$. In $U^{(l)}$, there are  ${N\choose 2}-{N-\rho_{l}\choose 2}=\mathcal{O}(\rho_l N)$ Givens rotations, which can be implemented in depth $\mathcal{O}(N)$ with Fermi-swap networks in qubits~\cite{Kivlichan2018}. Therefore, the total depth is $\mathcal{O}(LN)$. This is already optimal (at least without ancillas) and therefore, fermionic quantum computers cannot improve on the asymptotic scaling (but may improve the prefactor).

However, for large molecules requiring a large basis set $N$, $\rho_l=\mathcal{O}(1)$ (and $L=\mathcal{O}(N)$)~\cite{motta2019}. Therefore, the depth for $V^\ell$ is $\mathcal{O}(1)$. Furthermore, the rotation $\tilde U^{(\ell)}$ from one $\ell$ to another only requires $\mathcal{O}(\rho_l \rho_{l+1})=\mathcal{O}(1)$ Givens rotations~\cite{motta2021}. For qubits using a Fermi-swap network, $\tilde U^{(\ell)}$ would still have depth $\mathcal{O}(N)$. By contrast, for fermions, the $\mathcal{O}(1)$ Givens rotations (which are essentially hopping gates) only require depth $\mathcal{O}(1)$! Thus, $e^{i \Delta t V^{(\ell)}} \tilde{U}^{(\ell)}$ only requires depth $\mathcal{O}(1)$ (and also the same number of gates) for fermionic quantum computers, and, therefore, the overall Trotter step requires at most depth $\mathcal{O}(L)$. This means that fermionic quantum computers can simulate large molecules with linear depth whereas the Fermi-swap network qubit implementation takes quadratic depth.

In fact, the observation that $\prod_{\ell=1}^L e^{i \Delta t V^{(\ell)}} \tilde{U}^{(\ell)}$ only requires  $\mathcal{O}(L)$ gates means that likely the depth can be reduced further. In case the $V^{(\ell)}$ act on disjunct sets of sites, only $\Omega(1)$ depth is required. Of course, $e^{i \Delta t(h+S)} U^{(1)}$ is still restricted to depth $\mathcal{O}(N)$ just because of parallelization restriction, but as we discussed in the context of materials simulation, this restriction can be lifted by introducing $\mathcal{O}(N^2)$ ancillas~\cite{hoyer2005}, leading to depth $\mathcal{O}(\log(N))$. From these observations, we conjecture that a sub-linear depth for large molecules might be possible in practice in fermionic quantum computers. Testing this conjecture would require examining the structure of the $V^{(\ell)}$ for molecules of interest.

\end{document}